\documentclass[runningheads]{llncs}

\usepackage{graphicx}
\usepackage{hyperref}
\usepackage[T1]{fontenc}

\usepackage{mathrsfs}
\usepackage{xcolor}
\usepackage[draft]{todonotes}
\usepackage{amsmath}
\DeclareMathOperator*{\argmax}{arg\,max}
\usepackage[tableposition=top]{caption}
\DeclareCaptionLabelFormat{andtable}{#1~#2  \&  \tablename~\thetable}

\usepackage{algorithm,needspace}
\usepackage{algcompatible}
\usepackage[noend]{algpseudocode}
\MakeRobust{\Call} 

\algrenewcommand\algorithmicrequire{\textbf{Requisito:}}
\algrenewcommand\algorithmicensure{\textbf{Observação:}}
\algrenewcommand\algorithmicfunction{\textbf{function}}
\algrenewcommand\algorithmicfor{\textbf{for}}
\algrenewcommand\algorithmicdo{\textbf{do}}
\algrenewcommand\algorithmicend{\textbf{end}}
\algrenewcommand\algorithmicreturn{\textbf{return}}
\algrenewcommand\algorithmicprocedure{\textbf{procedure}}
\algrenewcommand\algorithmicif{\textbf{if}}
\algrenewcommand\algorithmicelse{\textbf{else}}
\algrenewcommand\algorithmicthen{\textbf{then}}
\algrenewcommand\algorithmicwhile{\textbf{while}}
\algdef{SE}[DOWHILE]{Do}{doWhile}{\algorithmicdo}[1]{\algorithmicwhile\ #1}%

\makeatletter
\renewcommand*{\ALG@name}{Algorithm}

\makeatother

\usepackage{support-caption}
\usepackage[numbers,sort]{natbib}
\usepackage{comment}
\usepackage{enumerate}
\usepackage{subfig}
\usepackage{lipsum}
\usepackage{svg}

\usepackage{amsmath}
\usepackage{amssymb}
\usepackage{makecell}
\usepackage[title]{appendix}

\usepackage{multirow}
\usepackage{booktabs}
\usepackage{pbox}

\usepackage{complexity}
\newclass{\Hard}{hard}
\newclass{\pNP}{paraNP}
\newclass{\Hness}{hardness}
\newcommand{\NPH}{\NP\text{-}\Hard}

\newcommand{\NPHness}{\NP\text{-}\Hness}
\newclass{\Complete}{complete}
\newclass{\Cness}{completeness}
\newcommand{\NPc}{\NP-\Complete}

\newcommand{\NPcness}{\NP-\Cness}
\newfunc{\YES}{YES}
\newfunc{\NOi}{NO}
\newfunc{\tw}{tw}
\newfunc{\sift}{ref}
\newcommand{\pname}[1]{\textsc{#1}}

\newcommand{\bigO}[1]{{\mathcal{O}\!\left(#1\right)}}
\newcommand{\td}[1]{{\mathbb{#1}}}

\newfunc{\opt}{opt}
\newfunc{\rmcx}{rmc}
\newcommand{\union}{\uplus}
\newcommand{\bigunion}{\biguplus}
\newfunc{\ins}{ins}
\newfunc{\shift}{shift}
\newfunc{\glue}{glue}
\newfunc{\proj}{proj}
\newfunc{\joinf}{join}

\newcommand{\join}{\sqcup}
\newcommand{\meet}{\sqcap}
\newcommand{\rmc}[1]{\rmcx\left(#1\right)}

\newcommand{\nproblem}[3]{{\centering\fbox{\pbox{\textwidth}{\pname{#1}\\\textit{Instance}: #2\\\textit{Question}: #3}}}}

\usepackage{todonotes}

\begin{document}
%


\title{Weighted Connected Matchings}

%
%
\author{Guilherme C. M. Gomes\inst{1}\and\\
Bruno P. Masquio\inst{2} \and\\
Paulo E. D. Pinto\inst{2} \and\\ Vinicius F. dos Santos\inst{1}\thanks{Partially supported by FAPEMIG and CNPq} \and\\ Jayme L. Szwarcfiter\inst{2,3}\thanks{Partially supported by FAPERJ and CNPq}}
\authorrunning{Gomes et al.}
%
\institute{
  Universidade Federal de Minas Gerais (UFMG) -- Belo Horizonte, MG -- Brazil 
  \email{\{gcm.gomes,viniciussantos\}@dcc.ufmg.br}
  \and
  Universidade do Estado do Rio de Janeiro (UERJ) -- Rio de Janeiro, RJ -- Brazil
  \email{\{brunomasquio,pauloedp\}@ime.uerj.br} \and
  Universidade Federal do Rio de Janeiro (UFRJ) -- Rio de Janeiro, RJ -- Brazil \email{jayme@nce.ufrj.br} \\
}
\maketitle              
\begin{abstract}
A matching $M$ is a $\mathscr{P}$-matching if the subgraph induced by the endpoints of the edges of $M$ satisfies property $\mathscr{P}$. As examples, for appropriate choices of $\mathscr{P}$, the problems \pname{Induced Matching}, \pname{Uniquely Restricted Matching}, \pname{Connected Matching} and \pname{Disconnected Matching} arise.
For many of these problems, finding a maximum $\mathscr{P}$-matching is a knowingly {\NPH} problem, with few exceptions, such as connected matchings, which has the same time complexity as the usual \pname{Maximum Matching} problem.
The weighted variant of \pname{Maximum Matching} has been studied for decades, with many applications, including the well-known \pname{Assignment} problem.
Motivated by this fact, in addition to some recent researches in weighted versions of acyclic and induced matchings, we study the \pname{Maximum Weight Connected Matching}.
In this problem, we want to find a matching $M$ such that the endpoint vertices of its edges induce a connected subgraph and the sum of the edge weights of $M$ is maximum.
Unlike the unweighted \pname{Connected Matching} problem, which is in {\P} for general graphs, we show that \pname{Maximum Weight Connected Matching} is {\NPH} even for bounded diameter bipartite graphs, starlike graphs, planar bipartite, and bounded degree planar graphs, while solvable in linear time for trees and subcubic graphs.
When we restrict edge weights to be non negative only, we show that the problem turns to be polynomially solvable for chordal graphs, while it remains {\NPH} for most of the cases when weights can be negative.
Our final contributions are on parameterized complexity.
On the positive side, we present a single exponential time algorithm  when parameterized by treewidth.
In terms of kernelization, we show that, even when restricted to binary weights, \pname{Weighted Connected Matching} does not admit a polynomial kernel when parameterized by vertex cover under standard complexity-theoretical hypotheses.

\end{abstract}

\section{Introduction}

The problems involving matchings have a vast literature in both structural and algorithmic graph theory~\cite{edmonds_matching,10.1007/s00453-003-1035-4,LOZIN20027,BrunoMasquio:2019:TeseMestrado,vazirani,MOSER2009715,10.1007/978-3-030-48966-3_31}.
A matching is a subset $M \subseteq E$ of edges of a graph $G = (V,E)$ that do not share any endpoint.
A $\mathscr{P}$-matching is a matching such that $G[M]$, the subgraph of $G$ induced by the endpoints of edges of $M$, satisfies property $\mathscr{P}$.
The problem of deciding whether or not a graph admits a $\mathscr{P}$-matching of a given size has been investigated for many different properties $\mathscr{P}$ over the years. One of the most well known examples is the {\NPcness} even for bipartite of the \pname{Induced Matching}, where $\mathscr{P}$ is being $1$-regular~\cite{CAMERON198997}. Other {\NPH} problems include \pname{Acyclic Matching}~\cite{GODDARD2005129}, \pname{$k$-Degenerate Matching}~\cite{BASTE201838}, \pname{Uniquely Restricted Matching}~\cite{Golumbic2001}, and \pname{Disconnected Matching}. For the latter, we prove its {\NPHness} even for bipartite and chordal graphs in~\cite{disconnected_matchings}. One of the few exceptions of a $\mathscr{P}$-matching problem polynomially solvable is \pname{Connected Matcihng}, in which $G[M]$ has to be connected.

It is worth mentioning that the name \pname{Connected Matching} was also used for another problem where it is asked to find a matching $M$ such that every pair of edges of $M$ has a common adjacent edge~\cite{cameron_connected_matching}. However, we adopt the more recent meaning of \pname{Connected Matching}, given by Goddard et. al~\cite{GODDARD2005129}, who also proved that the sizes of a maximum matching and a maximum connected matching in a connected graph are the same. We strengthen this result in~\cite{disconnected_matchings}, showing that a maximum connected matching can be obtained with the same complexity as a maximum matching.

Recently, some $\mathscr{P}$-matchings concepts were extended to edge-weighted problems, in which, in addition to the matching to have a certain property $\mathscr{P}$, the sum of the weights of the matched edges is sufficiently large. It was shown that \pname{Maximum Weight Induced Matching} can be solved in linear time for convex bipartite graphs~\cite{klemz2022} and in polynomial time for circular-convex and triad-convex bipartite graphs~\cite{panda2020}. Also, \pname{Maximum Weight Acyclic Matching} was approached in~\cite{dieter2019}, showing that the problem is polynomially solvable for $P_4$-free graphs and $2P_3$-free graphs.

Motivated by these studies, in addition to the pertinence in {\P} of the unweighted version of the problem, \pname{Maximum Weight Matching}, we study problems where the graphs are edge-weighted. Thereby, we present the decision \pname{Weighted Connected Matching} and the optimization \pname{Maximum Weight Connected Matching} problems, abbreviated respectively by \pname{WCM} and \pname{MWCM}, which we formally define as follows:

\nproblem{Weighted Connected Matching(WCM)}{An edge weighted graph $G$ and an integer $k$.}{Is there a matching $M$ whose sum of its edge weights is at least $k$ and such that $G[M]$ is connected?}

\nproblem{Maximum Weight Connected Matching(MWCM)}{An edge weighted graph $G$.}{Which matching $M$ of $G$ has the maximum sum of its edge weights?}

Clearly, we can see that \pname{WCM} is in {\NP}, as we show in the next proposition.

\begin{proposition}\label{prop:wcm-np}
\pname{Weighted Connected Matching} is in {\NP}
\end{proposition}

In some cases, we approach separately \pname{Weighted Connected Matching} when negative weights are allowed or not, denoting, respectively as \pname{WCM} and \pname{WCM}$^+$. Note that, unlike most of the weighted matching problems, in connected matchings, it is relevant to consider graphs having negative edge weights also. 

Our results include the {\NPcness} of \pname{WCM}$^+$ for bounded diameter bipartite graphs and planar bipartite graphs. For the more general problem \pname{WCM}, we show that it can be solved in linear time for trees or subcubic graphs, while is {\NPc} for bounded degree planar graphs and starlike graphs. Unlikely, for the latter class, when we restrict the weights to be non-negative only, as in \pname{WCM}$^+$, the problem turns to be in {\P}, as we prove so for chordal graphs. Finally, we give a single exponential algorithm parameterizing by treewidth and show that \pname{WCM}$^+$ does not admit a polynomial kernel when parameterized by vertex cover. We summarize some of these results in Table~\ref{tab:p-matchings}.

\begin{table}[ht]
    \centering
\begin{tabular}{p{0.25\textwidth}|p{0.20\textwidth}|p{0.2\textwidth}|p{0.2\textwidth}} 
    \hline
    \multicolumn{2}{c|}{\multirow{2}{*}{Graph class}} & \multicolumn{2}{c}{Complexity} \\ \cline{3-4}
 \multicolumn{2}{c|}{} & \makecell{\pname{WCM}$^+$} & \makecell{\pname{WCM}} \\ \hline
 
 \multicolumn{2}{c|}{\makecell{General}}   & \multicolumn{2}{c}{\multirow{2}{*}{\makecell{{\NPc}\\ (Theorem~\ref{teo:wcm-bip})}}} \\ \cline{1-2}

\multicolumn{1}{c|}{\multirow{1}{*}{\makecell{Bipartite\text{ }}}} & \makecell{diameter at \\ most $4$} & \multicolumn{2}{c}{} \\ \cline{1-2}\cline{3-4}


\multicolumn{2}{c|}{\multirow{1}{*}{\makecell{Chordal}}}   & \multicolumn{1}{c|}{\multirow{2}{*}{\makecell{{\P}\\ (Theorem~\ref{teo:wcm-chordal-poly})}}} & 
\multicolumn{1}{c}{\multirow{2}{*}{\makecell{{\NPc}\\  (Theorem~\ref{teo:wcm-starlike})}}} \\ \cline{1-2}

\multicolumn{2}{c|}{\makecell{Starlike}}  & \multicolumn{1}{c|}{} & \multicolumn{1}{c}{} \\\hline

\multicolumn{1}{c|}{\multirow{2}{*}{\makecell{\\ Planar}}} & \makecell{bipartite} &  \multicolumn{2}{c}{\makecell{\NPc \\ (Theorem~\ref{theo:wcm-planar-nn})}} \\ \cline{2-2}\cline{3-4}

\multicolumn{1}{c|}{} & \makecell{$\Delta \geq 3$} & \makecell{?} & \makecell{{\NPc}\\ (Theorem~\ref{teo:wcm-planar})} \\ \hline

\multicolumn{2}{c|}{\makecell{$\Delta < 3$}} & \multicolumn{2}{c}{\makecell{{\P} \\ (Theorem~\ref{wcm-degree-3})}} \\ \hline

\multicolumn{2}{c|}{\makecell{Tree}} & \multicolumn{2}{c}{\makecell{{\P}\\ (Theorem~\ref{teo:wcm-tree-linear})}} \\ \hline

\end{tabular}
\caption{Complexities for \pname{Weighted Connected Matching} admitting negative weights(\pname{WCM}) or not(\pname{WCM$^+$}).}. \label{tab:p-matchings}
\end{table}

\noindent \textbf{Preliminaries}. For an integer $k$, we define $[k] = \{1, \dots, k\}$.
For parameterized complexity, we refer to~\cite{cygan_parameterized}.
We use standard graph theory notation and nomenclature as in~\cite{murty,classes_survey}.
Let $G = (V, E)$ be a graph, $W \subseteq V$, $M \subseteq E$, and $V(M)$ to be the set of endpoints of edges of $M$, which are also called $M$-saturated vertices, or just saturated. Let $\Delta(G)$ be the maximum vertex degree of $G$. We denote by $G[W]$ the subgraph of $G$ induced by $W$; in an abuse of notation, we define $G[M] = G[V(M)]$.
A matching is said to be maximum and maximum weight if there is no other matching of $G$ with greater cardinality and sum of edge weights, respectively. A matching is perfect if $V(M) = V(G)$. 
Also, $M$ is said to be connected if $G[M]$ is connected. Let $uv$ be an edge of $G$. We denote $w(uv)$ by the weight of the edge $uv$ and $w(M)$ by $\sum_{uv \in M} w(uv)$. The operations $G-uv$ and $G-v$ result, respectively, the graphs $G'=(V,E\setminus\{uv\})$ and $G[V\setminus\{v\}]$. We denote $K_{i,j}$ by a bipartite complete graph whose bipartition cardinalities are $i$ and $j$. A star graph is a graph isomorphic to $K_{1,i}$, for some $i$. The graphs $P_n$ and $C_n$ are path and cycle graphs having $n$ vertices.
A graph $G$ is $H$-free if $G$ has no copy of $H$ as an induced subgraph; $G$ is chordal if it has no induced cycle with more than three edges. A clique tree is a tree $T$ representing a chordal graph $G$ in which vertices and edges of $T$ correspond, respectively, to maximal cliques and minimal separators of $G$. A graph is planar if it can be embedded in the plane without edge crossings.
A graph is a starlike graph if it is chordal and has a clique tree that is also a star graph. A boolean formula is monotone if, for each of its clauses, literals are either all positive or all negative.

This paper is organized as follows.
In Section~\ref{sec:wcm}, we show that \pname{Weighted Connected Matching} is {\NPc} for starlike and bounded vertex degree, while polynomially solvable for trees and subcubic graphs. In Section~\ref{sec:wcm+}, we show that the problem remains {\NPc} for bounded diameter bipartite graphs, planar bipartite, while is in {\P} for chordal graphs. In Section~\ref{sec:kernel}, we show that \pname{Weighted Connected Matching} on parameterized by vertex cover does not admit a polynomial kernel, even if the input is restricted to bipartite graphs of bounded diameter. In Section~\ref{sec:single_exp}, we present a single exponential algorithm when parametrized by treewidth.
Finally, we present our concluding remarks and directions for future work in Section~\ref{sec:conclusions}.

\section{Weighted Connected Matching in some graph classes}
\label{sec:wcm}

\subsection{Starlike}

In this section, we prove that \pname{Weighted Connected Matching} is {\NPc} even for starlike graphs having edge weights in $\{ -1, +1 \}$.
A graph is said to be starlike if it is chordal and its clique tree is a star.
For the reduction, we use the {\NPc} problem \pname{3SAT}\cite{garey_johnson}, in which we are given a set $C$ of clauses with exactly three literals each. The question is if there is a truth assignment of the variables of $C$ such that at least one literal of each clause resolves to true. For an instance, we also denote $X$ as the set of variables of $C$.

For the \pname{Weighted Connected Matching} input, we use $k=|X| + |C|$ and the following reduction graph $G_{X,C}$.

\begin{enumerate}[(I)]
    \item For each variable $x_i \in X$, add a copy of $C_3$ whose vertices are labeled $x_i$, $x_i^+$ and $x_i^-$. Set weight $-1$ to the $x_i^+x_i^-$ and $+1$ to the other edges.
    \item For each pair of variables $x_i,x_j \in X$, add all possible edges between vertices of $\{ x_i^-,x_i^+ \}$ and $\{ x_j^-,x_j^+ \}$ and set its weights to $-1$.
    \item For each clause $c_i \in C$, add a copy of $K_2$ whose edge weight is $+1$ and label its endpoints as $c_i^+$, and $c_i^-$. Also, for each literal $x_j$ of $c_i$, connect by a weight $-1$ edge $c_i^-$ and $c_i^+$ to $x_j^-$ if $x$ is negated, or $x_j^+$ otherwise.
\end{enumerate}

This graph is indeed starlike, as its clique tree is a star, having as center the maximal clique containing the vertices $\{ x_i,x_i^+, x_i^- \mid x_i \in X \}$.

\begin{lemma}\label{lemma:wcm-star-ida}
Given a solution $R$ for the \pname{3SAT} instance $(X,C)$, we can obtain a connected matching $M$ in the $G_{X,C}$ having weight $|X|+|C|$.
\end{lemma}
\begin{proof}
We show how to obtain the matching $M$. (i) For each clause $c_i \in C$, add the edge $c_i^-c_i^+$ to $M$. Also, (ii) for each variable $x_i \in X$, if $x_i = T$, we saturate the edge $x_i^+x_i$; otherwise, $x_i^-x_i$.

This matching is connected. Edges from (ii) are connected as they induce a clique. Each edge from (i), obtained by clause $c_i \in C$, having $x_j$ as the variable related to the literal that resolves to true in $c_j$, is connected. This holds because, if $x_j$ is negated, then $c^+_ix_j^- \in E(G_{X,C})$ and $x_j^-$ is saturated. Otherwise, $c^+_ix_j^+ \in E(G_{X,C})$ and $x_j^+$ is saturated.
\end{proof}

\begin{lemma}\label{lemma:wcm-star-volta}
Given an input $(X,C)$ for \pname{3SAT} and a connected matching $M$ in $G_{X,C}$ having weight $|X|+|C|$, we can obtain an assignment $R$ of $X$ that solves \pname{3SAT}
\end{lemma}
\begin{proof}
Denote $W_{-1}$ and $W_1$ as the edge sets from $G_{X,C}$ whose weights are, respectively, $-1$ and $1$.

First, we show that a matching having weight $|X|+|C|$ contains exactly $|X|+|C|$ edges from $W_1$ and no edges from $W_{-1}$.

Note that there can be at most $|X|+|C|$ edges from $W_1$. This holds because, for each variable $x_i \in X$, there is at most one saturated edge of $\{x_i^+x_i, x_i^-x_i\}$, since both have an endpoint in vertex $x_i$. Also, for each clause $c_i \in C$, the edge $c_i^-c_i^+$ can be saturated simultaneously.

Since all remaining vertices,  contained in $W_{-1}$, have negative weights, if there is a matching with $|X|+|C|$ vertices from $W_1$ and no vertices from $W_{-1}$, then it is maximum.

Therefore, if $M$ is a matching whose weight is $|X| + |C|$, then $|M \cap W_1| = |X|+|C|$ and $|M \cap W_{-1}| =0$.

Moreover, $M$ is connected, then, for each saturated edge $c_i^+c_i^-$, there is a saturated adjacent vertex, either $x_j^+$ or $x_j^-$, $x_j \in X$, $x_j \in c_i$. Those vertices are exactly the ones representing a literal in the clause $c_i$.

So, to obtain $R$, for each variable $x_i \in X$, we set $x_i = T$ if and only if $x_i^+$ is saturated.
\end{proof} 

\begin{theorem}\label{teo:wcm-starlike}
\pname{Weighted Connected Matching} is {\NPc} even for starlike graphs whose edge weights are in $\{ -1,+1 \}$.
\end{theorem}
\begin{proof}
Proposition~\ref{prop:wcm-np} shows that the problem belongs to {\NP}. According to the transformations between \pname{Weighted Connected Matching} and \pname{3SAT} solutions described in Lemmas \ref{lemma:wcm-star-ida} and \ref{lemma:wcm-star-volta}, the {\sc 3SAT} problem, which is {\NPc}, can be reduced to \pname{Weighted Connected Matching} using a starlike graph whose edge weights are either $-1$ or $+1$. Therefore, \pname{Weighted Connected Matching} is {\NPc} even for starlike graphs whose weights are in $\{ -1, +1\}$.
\end{proof}


\subsubsection{Example}

Consider an input of \pname{3SAT} defined by $B = (x_1 \vee \overline{x_2} \vee \overline{x_4}) \wedge (x_1 \vee \overline{x_3} \vee x_5) \wedge (\overline{x_1} \vee \overline{x_2} \vee x_4) \wedge (x_2 \vee x_3 \vee x_5)$.

In this example, the reduction graph used in \pname{Weighted Connected Matching} input is illustrated in Figure~\ref{fig:wcm-starlike-example-matching}, as well as a connected matching having weight $9$. Dashed and solid edges represent weight $-1$ and $1$, respectively, and the vertices in the dashed rectangle induce a clique in which the omitted edges have weight $-1$.

The illustrated matching corresponds to the assignment $(F,T,F,F,T)$ of the variables $(x_1, x_2, x_3, x_4, x_5)$ in $B$, in this order.

\begin{figure}
    \centering
    \includegraphics[width=1\textwidth]{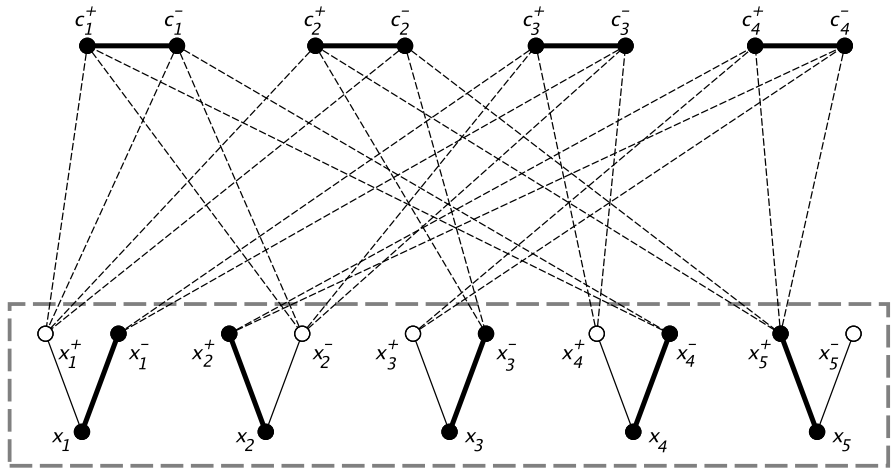}
    \caption{Example of a \pname{3SAT} reduction using a starlike graph}
    \label{fig:wcm-starlike-example-matching}
\end{figure}
\subsection{Planar graphs}

In this section, we prove the {\NPcness} of \pname{Weighted Connected Matching} even for planar graphs with maximum vertex degree three and edge weights in $\{-1,+1\}$. Our proof is an adaptation of the \pname{Weighted Connected Subgraph} made by Marzio De Biasi in~\cite{marzio_de_biasi}.

For this purpose, we use one of  Karp's original $21$ {\NPc} problems, the \pname{Steiner Tree}~\cite{karp_72}. In this problem, we are given a graph $G' = (V',E')$, a subset $R\subseteq V'$ and an integer $k' > 0$ and we want to know if there is a subgraph $T = (V_T, E_T)$ of $G'$ such that $T$ is a tree, $R \subseteq V_T$ and $|E_T| \leq k'$. 

Garey and Johnson showed in~\cite{rectilinear_steiner_tree} that this problem is {\NPc} even for planar graphs. Thereby, we use our reduction using the fact that the input graph of \pname{Steiner Tree} is planar. Also from \cite{rectilinear_steiner_tree}, we use the technique of adding cycles whose lengths are greater than $\Delta(G)$ that was in the {\NPcness} proof of {\sc Vertex Cover } for planar graphs with maximum vertex degree.

Let $(G'=(V',E'), R, k')$ be an input of { \sc Steiner Tree} such that $G'$ is planar and the following values for $q$, $p$ and $r$.

\begin{align*}
&q = \Delta(G') \\
&p = q(|V'| - |R|) + 1 \\
&r = p|E'| + 1
\end{align*}

For our input for the reduction to \pname{Weighted Connected Matching}, we set $k = r|R| - pk'$ and the input
graph $G$ built with the following procedures.

\begin{enumerate}[(I)]
    \item For each vertex $w$ add a copy of a cycle graph with length $2r$ if $w \in R$, and $2q$ otherwise. If this length is less then $3$, instead, add a copy of a path graph with the same number of vertices as the intended cycle length. Set the weights of all these edges to $+1$. Besides, add the label for $|N(w)|$ of these vertices as $v_{wu}$, for each $u \in N(w)$. Denote this subgraph as $C_w$.

    \item For each edge $wu$ in $E$, generate a copy of $P_{2p}$ whose edge weights are $-1$ and make its terminal vertices disjointly adjacent to $v_{wu}$ and $v_{uw}$. Denote this subgraph as $P_{wu}$.
\end{enumerate}

Next, we analyze the planarity of $G$, showing that a planar embedding of $G'$ can be used to build a planar embedding of $G$. Note that the cycles in $G$ generated in (I) can be positioned in the same place as vertices in $G'$. Also, all paths in $G$ generated in (II) can be positioned along the edges as in $G'$.

Now, let's consider the maximum vertex degree of $G$. Note that every vertex in the cycles of (I) has degree $2$, except the ones connected to one vertex of a path from (II), having degree $3$. All vertices from (II) have degree $2$.

Thereby, we can enunciate the following proposition, which will strengthen our {\NPcness} proof in terms of the input graph properties.

\begin{proposition}
Graph $G$ is planar and $\Delta(G) \leq 3$.
\end{proposition}

Next, in Lemmas~\ref{lemma:wcm-planar-deg-3-ida} and~\ref{lemma:wcm-planar-deg-3-volta}, we show the correspondence between solutions of \pname{Steiner Tree} and \pname{Weighted Connected Matching}. Finally, Theorem~\ref{teo:wcm-planar} concludes our {\NPcness} proof.

\begin{lemma}\label{lemma:wcm-planar-deg-3-volta}
Let $(G',R,k')$ be an input of {\sc Steiner Tree }, whose solution is $T = (V_T,E_T)$, and $G$ be the transformation graph obtained from it. We can obtain, in polynomial time, a connected matching $M$ in $G$ such that $w(M) \geq k = r|R| - pk'$.
\end{lemma}
\begin{proof}
Let's build a connected matching having a size at least $k$.

For each vertex $w \in V_T$, we saturate $\frac{|V(C_w)|}{2} = r$ edges from $C_w$. Note that the sum of the weights of these edges is $r|R|$, since all vertices of $R$ are contained in $V_T$ and, for each cycle $C_u$, $u \in R$, we can saturate $r$ edges having weight $+1$.

Moreover, for each edge $uw \in E_T$, we saturate $\frac{|V(P_{uw})|}{2} = p$ edges from $P_{uw}$. The sum of these edge weights is at most $-pk'$, since, for each $P_{uw}$, we saturate $p$ edges having weight $-1$, and $k' \geq |E_T|$.

Next, we show that $M$ is connected. Note that $V(M) = \{V(C_w) \mid w \in V_T\} \cup \{ V(P_{uw}) \mid uw \in E_T \}$. So, for each $uw \in E_T$, the vertices $V(C_u) \cup V(C_w) \cup \{v_{uw},v_{wu}\}$ are saturated.

Therefore, $M$ is connected, $w(M) \geq r|R| - pk'$, and it can be obtained in polynomial time.
\end{proof}

\begin{lemma}\label{lemma:wcm-planar-deg-3-ida}
Let $(G' = (V',E'),R,k')$ be an input of \pname{Steiner Tree}, $G$ be the transformation graph obtained from it, and $M$ a connected matching in $G$ such that $w(M) \geq k = r|R| - pk'$. We can obtain a tree $T = (V_T, E_T)$ that is a solution of \pname{Steiner Tree} instance.
\end{lemma}
\begin{proof}
We show that, in order to $w(M) \leq k$, we have to saturate vertices of every $V(C_w)$, $w \in R$. Suppose that this is not true, because there is a vertex $u \in R$ such that $|V(M) \cap V(C_u)| = 0$. Then, the weight of $M$ is at most the maximum number of $+1$ saturated edges of $C = \{ C_w \mid u \neq w \in V' \}$. For each cycle $C_w \in C$ we can have $\frac{|C_w|}{2}$ saturated edges. So, $w(M) \leq \sum\limits_{w \in V'}\frac{|C_w|}{2}$, and we have the following equation.

\[
\sum\limits_{w \in V'\setminus\{u\}}\frac{|C_w|}{2} = \sum\limits_{w \in R\setminus\{u\}}\frac{|C_w|}{2} + \sum\limits_{w \in V'\setminus{R}}\frac{|C_w|}{2} = r(|R|-1) + q(|V'| - |R|)
\]

If $w(M)$ is at least $k$, then $r(|R|-1) + q(|V'| - |R|) \geq k = r|R| - pk'$.
\begin{align*}
r(|R|-1) + (p-1) &\geq r|R| \geq r|R| - pk' \\
p-1 &\geq r \\
p &\geq p|E'|+2
\end{align*}
This is a contradiction, which means that for every $w \in R$, at least one vertex of $V(C_w)$ is saturated, each having $r$ as the sum of edge weights. For $M$ to be connected, some $V(P_{uw})$ are all saturated, $uw \in E'$. The number of those paths that have all their vertices saturated is at most $k'$, since $w(M) \geq r|R| - pk'$. 

Note that saturating $V(C_w)$, $w \in V'\setminus{R}$, is irrelevant, because $p=q(|V|-|R|)+1$, and even if we saturate all those cycles, the maximum weight obtained is $q(|V|-|R|)$.

Thereby, we can build $T$ such that $V_T = \{ w \mid 0 < |V(C_w)\cap{V(M)|} \}$ and $E_T = \{ uw \mid V(P_{uw})\cap{V(M)} = V(P_{uw}) \}$. Note that $R \subseteq V_T$ and $|E_T| \leq k'$.
\end{proof}

\begin{theorem}\label{teo:wcm-planar}
\pname{Weighted Connected Matching} is {\NPc} even for planar graphs with maximum vertex degree $3$ and edge weights in $\{ -1, +1 \}$.
\end{theorem}
\begin{proof}
Proposition~\ref{prop:wcm-np} shows that the problem is in {\NP}. According to the transformations between \pname{Weighted Connected Matching} and \pname{Steiner Tree} solutions described in Lemmas~\ref{lemma:wcm-planar-deg-3-ida} and~\ref{lemma:wcm-planar-deg-3-volta}, the \pname{Steiner Tree} problem restricted to planar graphs, which is {\NPc}, can be reduced to \pname{Weighted Connected Matching} using a planar graph whose edge weights are either $-1$ or $+1$ and vertex degree is at most $3$. Therefore, \pname{Weighted Connected Matching} is {\NPc} even for planar graphs whose weights are in $\{ -1, +1\}$ and vertex degree is at most $3$.
\end{proof}

\subsection{Example 1}

We consider the input for {\sc Steiner Tree} as $k' = 1$, $R = \{ a,b \}$ and the graph $G' = (V', E')$ isomorphic to $C_3$, with $V' = \{ a,b,c \}$.

In this case, we have the following variable values.

\begin{align*}
&q = \Delta(G') = 2 \\
&p = q(|V| - |R|) + 1 = 3 \\
&r = p|E| + 1 = 10
\end{align*}

So, for the {\sc Weighted Connected Matching} input, we set $k = 17$, and the graph as the one illustrated in Figure~\ref{fig:wcm-planar-example-alt}.

\begin{figure}
    \centering
    \includegraphics[width=0.8\textwidth]{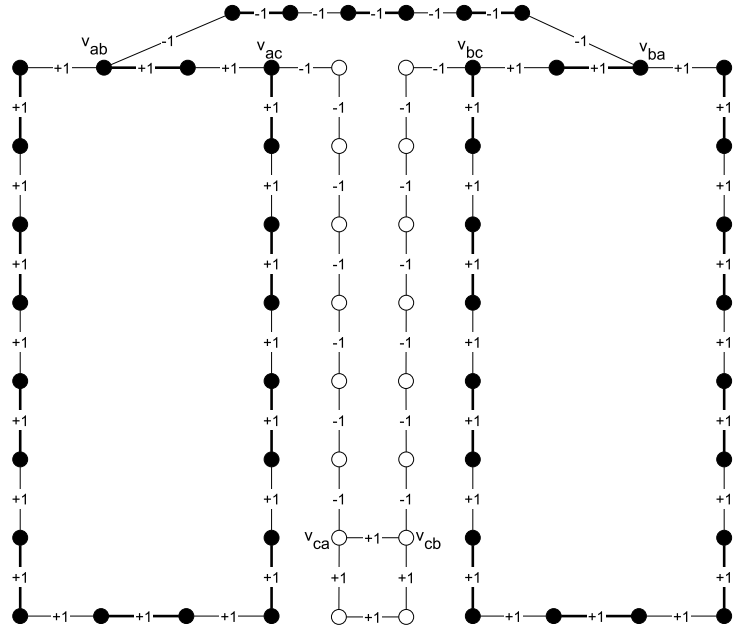}
    \caption{Example of a reduction graph $G$}
    \label{fig:wcm-planar-example-alt}
\end{figure}

To solve {\sc Weighted Connected Matching} for this instance, we need to find $M \subseteq E'$ such that $G[M]$ is connected and the sum of its edge weights is at least $17$.

To reach this value, $M$ must contain edges of each of the $C_{20}$ subgraphs, which were generated because of $a$ and $b$ of $G'$. Observe that it is possible to add $20$ edges of these subgraphs. For $G[M]$ to be connected, there must be a saturated path between those $C_{20}$ subgraphs. Note that there are only two possible paths of these, one with length $15$, between $v_{ac}$ and $v_{bc}$, and the other with length $7$, between $v_{ab}$ and $v_{ba}$. Note that optimally saturating disjoint edges from the first path would result in a matching having weight at most $16$ while, from the second, $17$.

Then, we can conclude that the only possible matching $M$ having weight $17$ is the one shown in Figure~\ref{fig:wcm-planar-example-alt} and it corresponds to the {\sc Steiner Tree} solution $G[\{a,b\}]$.

\subsection{Example 2}

We consider the input for {\sc Steiner Tree} as $k' = 3$, $R = \{ a,c,d \}$ and the graph $G'$ as in Figure~\ref{fig:wcs-red-2-1}. In this case, we have the following variable values.

\begin{align*}
&q = \Delta(G') = 4 \\
&p = q(|V| - |R|) + 1 = 13 \\
&r = p|E| + 1 = 105
\end{align*}
\begin{figure}
    \centering
    \includegraphics[width=0.5\textwidth]{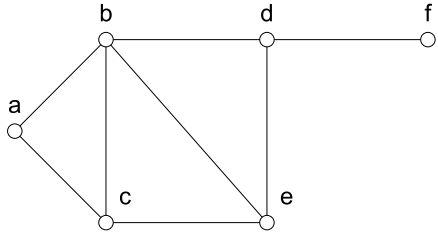}
    \caption{Example of a representation of a reduction graph $G'$ for {\sc Steiner Tree}}
    \label{fig:wcs-red-2-1}
\end{figure}

The transformation subgraph is shown in Figure~\ref{fig:wcs-red-2-2}. For  better visualization, we represent the path subgraphs $P_{2p}$ by a square. Each of its edges is incident to terminal vertices of $P_{2p}$. Also, a triangle represents a cycle subgraph, in which $\triangledown$ is a $C_{2q}$ and $\triangle$, a $C_{2r}$. Triangle edges represent edges that are incident to distinct vertices in its relative subgraph. Inside each of these symbols, we show the weight of a maximum matching of the corresponding subgraph.

\begin{figure}
    \centering
    \includegraphics[width=1\textwidth]{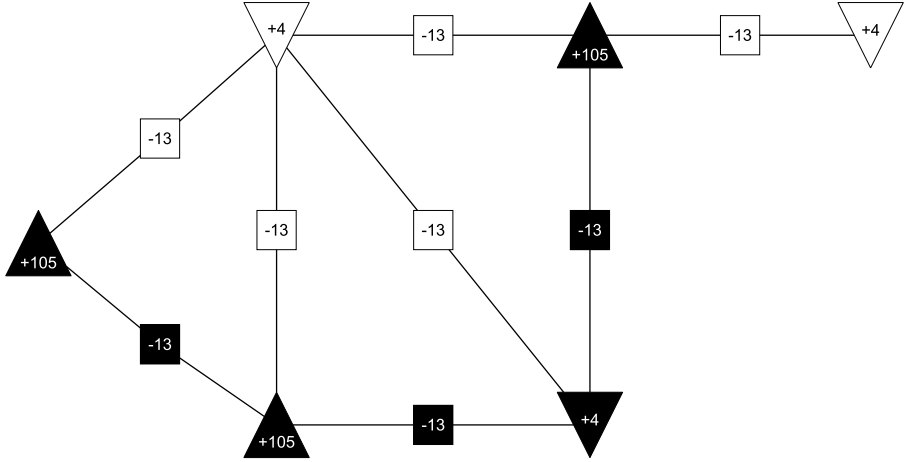}
    \caption{Representation of a reduction graph $G$ for {\sc Steiner Tree}}
    \label{fig:wcs-red-2-2}
\end{figure}

Note that the solution contains edges from $C_{210}$ cycles, as well as from at most $3$ path subgraphs $P_{26}$. So, there are three {\sc Steiner Tree} solutions, $G[\{ a,c,d,e \}]$, $(G[\{a,b,c,d\}] - ab)$ and $(G[\{a,b,c,d\}] - bc)$.

\subsection{Trees}\label{sec:wcm-trees}


In this section, we show that \pname{Maximum Weight Connected Matching} can be solved in linear time for trees. Despite this class having very strict properties, it is a good example as it is widely studied and has several applications.


We describe a linear algorithm that solves \pname{WCM} for trees. Let $T$ be a tree and $r$ any vertex of $V(T)$, which we call the root of $T$. Denote $S(r,v)$ as the set of children of $v$ in $T^r$. Also, consider $B_{r,v}$ as the weight of a maximum weight connected matching in $T^r_v$ such that, if $v$ is not a leaf, then $v$ is saturated with a vertex, denoted as $b_{r,v}$. Moreover, $\overline{B_{r,v}}$ is the weight of a matching $M$ defined as the union of the maximum connected matching in $T^r_u$, for each $u \in S(r,v)$, such that $M$ is connected if $v$ is saturated with his parent in $T^r$. In our algorithm, it is not required to calculate $\overline{B_{r,r}}$, which we do not define.

Next, we describe a dynamic programming algorithm that can be used to obtain those variables. We treat separately the base case, where the vertex analyzed is a leaf. For this case, we set all the variables to $0$.

\[
    B_{r,v} = \overline{B_{r,v}} = 0
\]

For the general case, where a vertex $v$ is not a leaf, we define the variables $B_{r,v}$, $\overline{B_{r,v}}$ and $b_{r,v}$ using a function $f$ in the following.


\[
f(r,vu) = \overline{B_{r,u}} + w(vu) + \sum_{\substack{s \in S(r,v) \setminus \{u \}}} \max\{B_{r,s},0\}
\]

\[
    \overline{B_{r,v}} = \sum_{u \in S(r,v)} \max\left\{B_{r,u},0\right\}
\]

\[
    B_{r,v} = \max\limits_{u \in S(r,v)} f(r,vu)
\]

\[
    b_{r,v} = \argmax\limits_{u \in S(r,v)} f(r,vu)
\]
 
Moreover, we denote $h$ as the vertex that maximizes $B_{r,h}$.
 
\[
    h = \argmax_{v \in V(T)} B_{r,v}
\]

Finally, given a vertex $v$ we recursively build a maximum weight connected matching in $T^r_v$.

In the following, we conclude this section, by showing the problem for trees is solvable in linear time.

\begin{theorem}\label{teo:wcm-tree-linear}
\pname{Maximum Weight Connected Matching} for trees can be solved in linear time.
\end{theorem}
\begin{proof}
For this proof, we describe a procedure how to obtain a maximum weight connected matching in a tree $T$. First, set $r$ as an arbitrary node of $T$. Next, we calculate $B_{r,v}$ for every $v \in V(T)$ using the dynamic programming technique. We calculate $B_{r,u}$ for the vertices $u$ of a postorder sequence tree search in $T^r$ starting on $r$.

The summations $\sum\limits_{u \in S(r,v)} \max\{B_{r,u},0\}$ for every vertex $v$ can be calculated in linear time. When we need this summation for every child of $v$ except $u$, in form of $\sum\limits_{s \in S(r,v)\setminus\{u\}} \max\{B_{r,s},0\}$, we just subtract $\max\{B_{r,u},0\}$ from the previously calculated summation. Thereby, the whole procedure can be done in linear time.

The matching can also be obtained in linear by the reconstruction of the dynamic programming we described.

\end{proof}

\subsection{Example}
As an example, consider as input the tree illustrated in Figure~\ref{fig:wcm-tree-example}. The vertex $a$ is chosen to be the root. In Table~\ref{tab:wcm-tree-example}, we show the related variables obtained. Also, the rows are ascending in the same order as those variables can be calculated using our dynamic programming.

In this example, the vertex $h$ that maximizes $B_{a,h}$ is $a$. So, we build the matching $M(a,a)$, which is defined in $T^a_a$. First, we saturate $a$ with $b_{a,a} = b$. Then we add the following partial matchings.
\begin{align*}
M(a,e) &= \{ eb_{a,e} \} = \{ ej \} \\
M(a,c) &= \{ cb_{a,c} \} = \{ cf \} \\
M(a,d) &= \{ \}
\end{align*}

Note that, though the subtree $T^a_d$ is not empty, there is no possible matching that can be added to increase the weight of a connected matching containing $a$. Then, the subtree $T^a_d$ is discarded, and we set $M(a,d) = \{ \}$.

Finally, we obtain the maximum weight connected matching $\{ab,ej,cf\}$, illustrated in Figure~\ref{fig:wcm-tree-example}, whose weight is $B_{a,a} = 12$.

 \begin{minipage}{\textwidth}

  \begin{minipage}[b]{0.39\textwidth}
    \centering
    \includegraphics[width=1\textwidth]{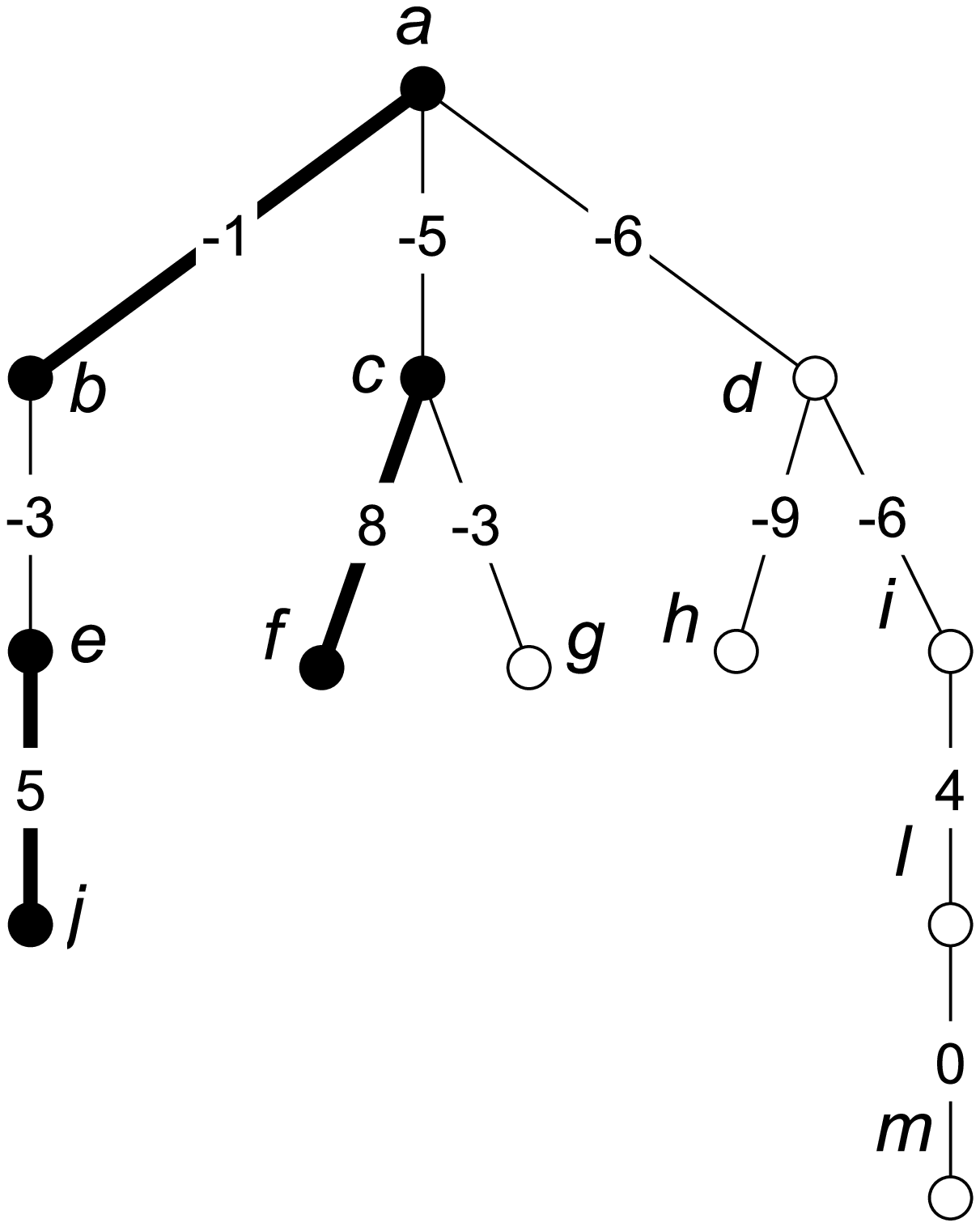}
    \captionof{figure}{A maximum weight connected matching in a tree}
    \label{fig:wcm-tree-example}
  \end{minipage}
  \hfill
  \begin{minipage}[b]{0.59\textwidth}
    \centering
    \begin{tabular}{c|c|c|c}
        $v$ & $B_{r,v}$ & $\overline{B_{r,v}}$ & $b_{r,v}$ \\ \hline
        j & 0 & 0 & - \\ \hline
        e & 5 & 0 & j \\ \hline
        b & -3 & 5& e \\ \hline
        f & 0 & 0 & - \\ \hline
        g & 0 & 0 & - \\ \hline
        c & 8 & 0 & f \\ \hline
        h & 0 & 0 & - \\ \hline
        m & 0 & 0 & - \\ \hline
        l & 0 & 0 & m \\ \hline
        i & 4 & 0 & l \\ \hline
        d & -5 & 4& h  \\ \hline
        a & 12 & 8& b  \\ \hline
    \end{tabular}
    \bigskip
    \label{tab:wcm-tree-example}
    \captionof{table}{Values of $B_{r,v}$, $\overline{B_{r,v}}$, and $b_{r,v}$ \\ for each vertex $v$}
    \end{minipage}
    
  \end{minipage}
\subsection{Graphs having degree at most $2$}

In this section, we prove that \pname{Weighted Connected Matching} can be solvable in linear time for graphs such that the maximum degree is at most $2$. Observe that a connected graph of this kind is either a cycle or a path. The algorithm for trees described in Section~\ref{sec:wcm-trees} works for path graphs also.

In the next theorem, we describe a linear algorithm for cycles and then conclude the complexity for graphs having a degree at most $2$ in Corollary~\ref{coro:wcm-degree-3}.

\begin{theorem}\label{wcm-degree-3}
\pname{Maximum Weight Connected Matching} can be solved in linear time for cycle graphs.
\end{theorem}
\begin{proof}

Let $e_1, \ldots, e_n$ be the edges of the cycle $C$. First, we will compute the maximum weight connected matching containing $w_n$. Let $s_i = \sum_{j = 1}^i w(e_{n-2j})$ for $i < \lfloor (n-1)/2 \rfloor$ and let $s_i' = \max_{1 \leq j \leq i} s_j$. 
Similarly, let $p_i = \sum_{j = 1}^i w(e_{2j})$ for $i < \lfloor (n-1)/2 \rfloor$. Finally, let $s_0 = p_0 = 0$. 
Note that a matching containing edge $e_n$ can be written as $\{e_2, \ldots, e_{2i}\} \cup \{e_n, \ldots, e_{n-2j}\}$, for some $i+j \leq n/2-1$, and its weight is given by $p_i + w(e_n) + s_j$. Also, note that, for  fixed $i$, the maximum matching among every possible $j$ is given by $S_i = p_i + w(e_n) + s_j'$, for $j = \lfloor n/2\rfloor - i - 1$. Since all $s_i, s_i'$ and $p_i$ can be precomputed  in linear time, we can compute each $S_i$ in constant time, which allows us to compute the maximum weight connected matching containing $e_n$ in linear time. 

Since $e_n$ was chosen arbitrarily, we can do the same procedure for a different edge, $e_{n-1}$. If neither $e_{n-1}$ nor $e_n$ belongs to a maximum weight connected matching, we know that $\{v\} = e_{n-1} \cap e_n$ is not saturated by it. Hence, we can delete $v$ from $C$ and compute the maximum weight connected matching of the resulting graph using the algorithm for trees from Section~\ref{sec:wcm-trees}.

Since any maximum weight connected matching contains either $e_{n-1}$, $e_n$ or none of them, if we take the maximum over these three cases, we have the maximum weight connected matching.

\end{proof}

\begin{corollary}\label{coro:wcm-degree-3}
\pname{Maximum Weight Connected Matching} can be solved in linear time for graphs having a degree at most $2$.
\end{corollary}

\section{Weighted Connected Matching with no negative weight in some graph classes }\label{sec:wcm+}

\subsection{Bipartite graphs}
\label{sec:bipartite_complexity}

In this section, we prove that \pname{Weighted Connected Matching} is {\NPc} even for bipartite graphs whose edge weights are in $\{ 0,1 \}$. We also approach the problem in terms of the diameter $d$ of the input graph. We show that when $d\geq4$, the problem is {\NPc}, while is in {\P} if $d\leq3$.

For this reduction, we use the {\NPc} problem \pname{3SAT}, and the input of \pname{Weighted Connected Matching} as $k = |X| + |C| + 1$ and the graph $G_{X,C}$ obtained by the following rules, based on the \pname{3SAT} input.

\begin{enumerate}[(I)]
    \item Add two vertices, $h^+$ and $h^-$, connected by a weight $1$ edge.
    \item For each variable $x_i \in X$, add a copy of $P_3$ whose edge weights are $1$, and label its endpoints as $x_i^+$ and $x_i^-$. Moreover, connect the other vertex, labeled $x_i$, to $h^+$ and set this edge weight to $0$.
    \item For each clause $c_i \in C$, add a copy of $K_2$ whose edge weight is $1$ and label its vertices as $c_i^+$ and $c_i^-$. Also, for each literal $x_j$ of $c_i$, add the edge $c_i^+x_j^-$ if $x_j$ is negated, or $c_i^+x_j^+$ otherwise.
\end{enumerate}

In the next lemmas, we show that, given an input $(X,C)$ for \pname{3SAT}, it is possible to obtain in linear time a connected matching $M$ in $G_{X,C}$, $w(M) = k$, if we have a solution for $C$, and vice versa. Denote $W_0$ and $W_1$ as the edge sets from $G_{X,C}$ whose weights are, respectively, $0$ and $1$.

\begin{lemma}\label{lemma:wcm-bip-ida}
Given a solution $R$ for the \pname{3SAT} instance $(X,C)$, we can obtain a connected matching $M$ having weight $|X|+|C|+1$ in $G_{X,C}$.
\end{lemma}
\begin{proof}
We show how to obtain the matching $M$. (i) For each clause $c_i \in C$, add the edge $c_i^-c_i^+$ to $M$. Also, (ii) for each variable $x_i \in X$, if $x_i = T$, we saturate the edge $x_i^+x_i$; otherwise, $x_i^-x_i$. Moreover, (iii) we saturate the edge $h^+h^-$.

We show that this matching is connected. Edges from (ii) are connected as they are connected to the edge $h^+h^-$ of (iii). Each edge from (i), obtained by clause $c_i \in C$, having $x_j$ as the variable related to a literal that resolves to true in $c_j$, is connected. This holds because, if $x_j$ is negated, then $c^+_ix_j^- \in E(G_{X,C})$ and $x_j^-$ is saturated. Otherwise, $c^+_ix_j^+ \in E(G_{X,C})$ and $x_j^+$ is saturated.
\end{proof}

\begin{lemma}\label{lemma:wcm-bip-volta}
Given an input $(X,C)$ for \pname{3SAT} and a connected matching $M$ in $G_{X,C}$ having weight $|X|+|C|+1$, we can obtain an assignment $R$ of $X$ that solves \pname{3SAT}
\end{lemma}
\begin{proof}
First, we show that a matching having weight $|X|+|C|+1$ contains exactly $|X|+|C|+1$ edges from $W_1$ and no edges from $W_0$.

Note that there can be at most $|X|+|C|+1$ edges from $W_1$. This holds because, for each variable $x_i \in X$, there is at most one saturated edge of $\{x_i^+x_i, x_i^-x_i\}$, since both have an endpoint in vertex $x_i$. Also, $h^+h^-$ can be saturated, as well as, for each clause $c_i \in C$, the edge $c_i^-c_i^+$.

Observe that each edge from $W_0$ has an endpoint in either $h^+$ or $S_1 = \{ c_i^+ \mid c_i \in C \}$. Saturating any of these vertices by a $W_0$ edge, will decrease the number of possibly saturated edges of $W_1$ and the weight of $M$, resulting $w(M) < k$. For instance, if we saturate $h^+$ or $c_i^+ \in S_1$, we will not be able to saturate, respectively, $h^+h^-$ or $c_i^+c_i^-$.

Therefore, $|M \cap W_1| = |X|+|C|+1$ and $|M \cap W_0| =0$.

If the matching $M$ is connected, then, for each saturated edge $c_i^+c_i^-$, there is a saturated adjacent vertex, either $x_j^+$ or $x_j^-$, $x_j \in X$, $x_j \in c_i$. Also, for each variable $x_i \in X$, the vertex $x_i$ is saturated, which is connected to the edge $h^+h^-$.

Hence, to obtain $R$, for each variable $x_i \in X$, we set $x_i = T$ if and only if $x_i^+$ is saturated. 
\end{proof} 

\subsection{Diameter $4$ bipartite graphs}

It is also possible to add the following rule to the \pname{Weighted Connected Matching} input graph $G_{X,C}$ in order to strengthen our {\NPcness} proof in terms of the graph diameter.

\begin{enumerate}[(IV)]
    \item Add a vertex $u$ and the weight $0$ edges defined by $\{ c_i^+u \mid c_i \in C \} \cup \{ x_i^+h^-, x_i^-h^- \mid x_i \in X \} \cup \{ x_ix_j^+, x_ix_j^-, x_jx_i^+, x_jx_i^- \mid x_i,x_j \in X \}$.
\end{enumerate}

For this graph, the previous lemmas are still valid by similar arguments. Thus, this graph can also be used for the {\NPcness} reduction.

In Figure~\ref{fig:wcm-bipartite-diam4-example-matching}, we illustrate an example with a full reduction diameter four bipartite graph, generated from the \pname{3SAT} input $B = (x_1 \vee \overline{x_2} \vee \overline{x_4}) \wedge (x_1 \vee \overline{x_3} \vee \overline{x_4}) \wedge (\overline{x_1} \vee \overline{x_2} \vee x_4) \wedge (x_2 \vee \overline{x_3} \vee \overline{x_4})$. Dashed and solid edges represent, respectively, edges having weights $0$ and $1$.

\begin{figure}
    \centering
    \includegraphics[width=1\textwidth]{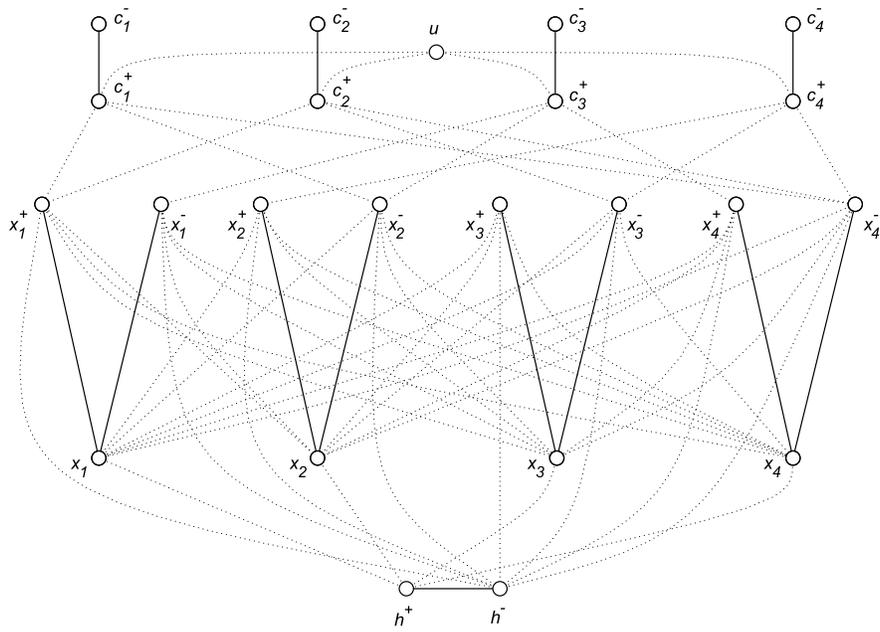}
    \caption{Example of a reduction diameter $4$ bipartite graph}
    \label{fig:wcm-bipartite-diam4-example-matching}
\end{figure}

Thus, we conclude the {\NPcness} with the following theorem.

\begin{theorem}\label{teo:wcm-bip}
{\sc Weighted Connected Matching} is {\NPc} even for bipartite graphs whose diameter is at most $4$, and edge weights are in $\{0,1\}$.
\end{theorem}
\begin{proof}
By Proposition~\ref{prop:wcm-np}, the problem is in {\NP}. According to the transformations between \pname{Weighted Connected Matching} and \pname{3SAT} solutions described in Lemmas \ref{lemma:wcm-bip-ida} and \ref{lemma:wcm-bip-volta}, the {\sc 3SAT} problem, which is {\NPc}, can be reduced to \pname{Weighted Connected Matching} using a bipartite graph whose diameter is at least $4$ and the edge weights are either $0$ or $1$. Therefore, \pname{Weighted Connected Matching} is {\NPc} even for bipartite graphs whose weights are in $\{ 0, 1\}$ and diameter is at most $d \geq 4$.
\end{proof}


\subsection{Example}

Consider an input of \pname{3SAT} defined by $B = (x_1 \vee \overline{x_2} \vee \overline{x_4}) \wedge (x_1 \vee \overline{x_3} \vee x_5) \wedge (\overline{x_1} \vee \overline{x_2} \vee x_4) \wedge (x_2 \vee x_3 \vee x_5)$.

In this example, the reduction graph used in \pname{Weighted Connected Matching}$^*$ input is illustrated in Figure~\ref{fig:wcm-pos-example-matching}, as well as a connected matching having weight $10$. Dashed and solid edges represent weights $0$ and $1$, respectively. We also omit the vertex $u$ and the edges from (IV).

The illustrated matching corresponds to the assignment $(F,T,F,F,T)$, in this order, of the variables $(x_1, x_2, x_3, x_4, x_5)$ of $B$.

\begin{figure}
    \centering
    \includegraphics[width=1\textwidth]{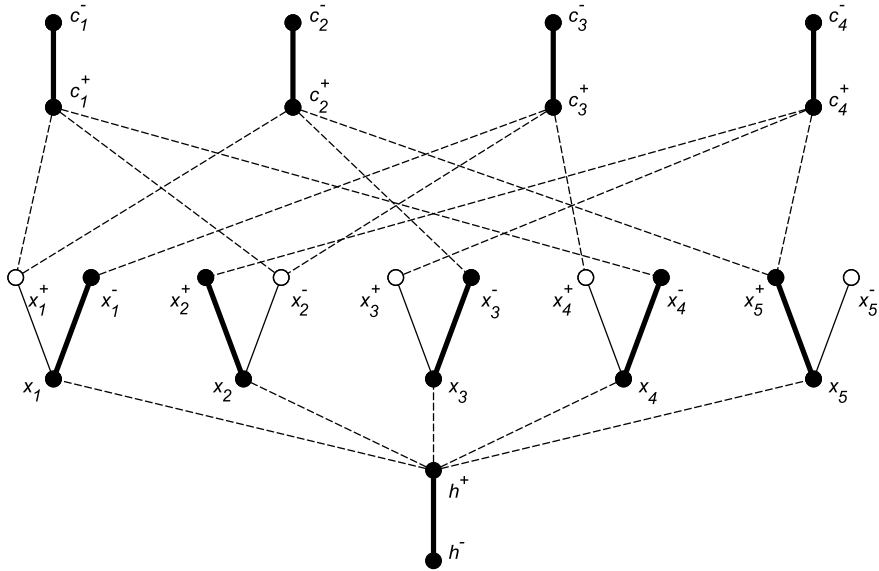}
    \caption{Example of a reduction bipartite graph for {\sc 3SAT}}
    \label{fig:wcm-pos-example-matching}
\end{figure}

\subsection{Planar bipartite graphs}

In this section we prove the {\NPcness} of \pname{Weighted Connected Matching} even for planar bipartite graphs having weights either $0$ or $1$. 

Let $B$ be a conjunctive formula such that $C = \{ c_1,\ldots,c_q \}$ and $X = \{ x_1,\ldots,x_m \}$ are the sets of clauses and variables of $B$, respectively. Also, let $G(B)$ be a graph in which there is a vertex for each clause or variable, and edges between a variable and a clause vertex if and only if the variable is part of the clause. Moreover, there is a cycle in all the variable vertices.

Observe an example of a graph $G(B)$ in Figure~\ref{fig:wcm-planar-nn-gb} referring to the formula $B = (x_1 \vee x_2 \vee x_5) \wedge (x_2 \vee x_3 \vee x_4) \wedge (\overline{x_2} \vee \overline{x_4} \vee \overline{x_5})$.

\begin{figure}
    \centering
    \includegraphics[width=0.8\textwidth]{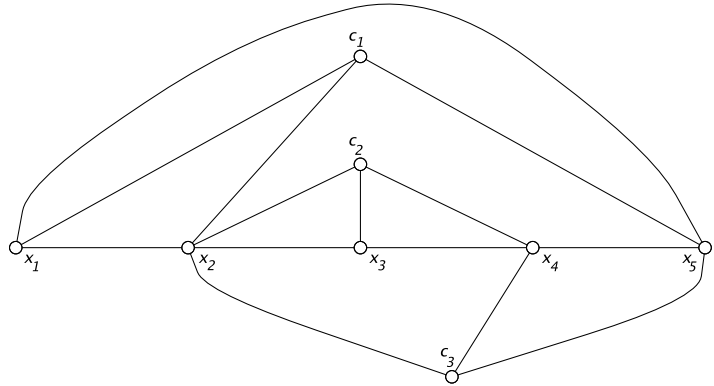}
    \caption{Example of a graph $G(B)$}
    \label{fig:wcm-planar-nn-gb}
\end{figure}

In~\cite{planar_formulae}, David Lichtenstein denoted a conjunctive formula $B$ planar if $G(B)$ is planar. In the same paper, it was studied a \pname{SAT} variant for which the input formulae are planar. This problem was denoted as \pname{Planar SAT}, which was proven to be {\NPc}.

Later in the paper, David Lichtenstein also proved the {\NPcness} of a \pname{Planar SAT} variant of our interest. In this problem, the input formula is monotone and, there is a planar $G(B)$ embedding such that each edge referencing a positive(negative) literal is connected to the top(bottom) of the variable vertex. Since this problem has not been given a name in its original paper, we will call it \pname{Planar Monotone SAT}.

We use \pname{Planar Monotone SAT} for our reduction from \pname{Weighted Connected Matching}, for which the input is defined as $k=2|X|+|C|$ and the planar bipartite graph, that we call $G'(B)$, is obtained by $G(B)$, with the addition of the following procedures.

\begin{enumerate}[(I)]
    \item For each variable $x_i$, generate four vertices, $x_i^+$, $x_i^-$, $v_i$, $u_i$. Also, add the weight $1$ edge $v_iu_i$.
    \item Remove all the edges from the variable vertices cycle and add the edges $\{ v_ix_i, v_ix_{i+1} \mid 1 \leq i < |X| \} \cup \{ v_{m}x_{m}, v_{m}x_1 \}$ having weight $0$.
    \item For each edge having an endpoint in $x_i$, $1 \leq i \leq |X|$, representing a positive(negative) literal, we connect  this edge to $x_i^+$($x_i^-$) instead of $x_i$ and set its weight to $1$.
    \item For each clause $c_i$, we rename the corresponding vertex to $c_i^+$ and connect it to a new vertex, $c_i^-$, by a weight $1$ edge.
\end{enumerate}

First, we show that the graph is planar and bipartite since this information will be used to strengthen our {\NPcness} proof.

\begin{lemma}
The graph $G'(B)$ is planar and bipartite.
\end{lemma}
\begin{proof}
Observe that $G'(B)$ is bipartite since one of its bipartitions can be defined as $\{ x_i,u_i \mid 1 \leq i \leq |X| \} \cup \{ c_i^+ \mid 1 \leq i \leq |C| \}$.

Next, we show that $G'(B)$ is planar. Note that the former graph $G(B)$ has all edges connecting positive(negative) literals are connected to the top(bottom) of the variable vertices.

For each variable $x_i$ since we created vertices $x_i^+$ and $x_i^-$ representing the literals $x_i$ and $\overline{x_i}$. Then, in the graph embedding, we can simply position $x_i^+$ and $x_i^-$, respectively, to the top and the bottom of $x_i$. Those literal variables do not cross the former variable cycle embedding, as we can see in Figure~\ref{fig:wcm-planar-var-gadget}.

\begin{figure}%
        \centering
        \subfloat[]{{\includegraphics[width=0.2\textwidth]{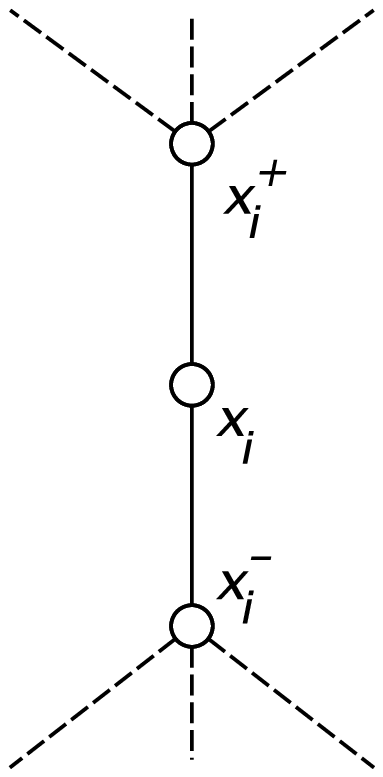}\label{fig:wcm-planar-var-gadget} }}
        \hspace{0.15\textwidth}
        \subfloat[]{{\includegraphics[width=0.6\textwidth]{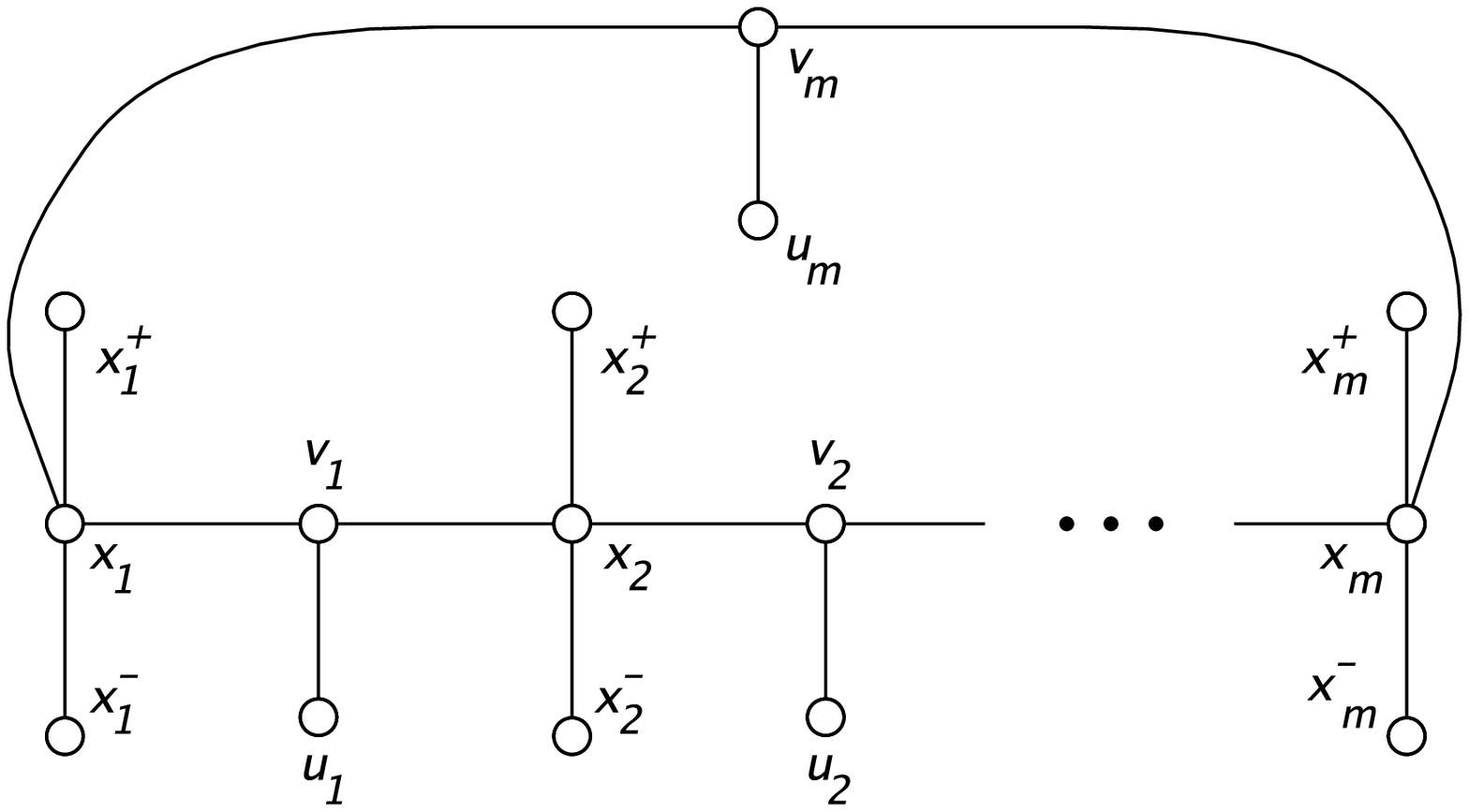}\label{fig:wcm-planar-var-cycle} }}
        \caption{Planar representations of some subgraphs of $G'(B)$}.%
        \label{fig:wcm-planar-gadgets}
\end{figure}

Also, the vertices $v_i$ and $u_i$ can be positioned along the edges of the variable vertices cycle, as in Figure~\ref{fig:wcm-planar-var-cycle}.

Therefore, this graph is bipartite and planar.
\end{proof}

In Lemmas~\ref{lemma:wcm-planar-nn-ida} and \ref{lemma:wcm-planar-nn-volta}, we show that, given an input $B$ for \pname{Planar Monotone SAT}, it is possible to obtain in linear time a connected matching $M$ in $G'(B)$, $w(M) = k$, if we have a solution for $B$, and vice versa. Then, Theorem~\ref{theo:wcm-planar-nn} concludes the {\NPc}ness proof.

\begin{lemma}\label{lemma:wcm-planar-nn-ida}
Given a planar formula $B$ and an assignment $R$ that resolves $B$ to true, we can obtain in linear time a connected matching $M$ in $G'(B)$ such that $w(M) = k = 2|X|+|C|$.
\end{lemma}
\begin{proof}
Let's build a connected matching $M$ having weight $k$. First, for each variable $x_i \in X$, we add to $M$ the edge $x_ix_i^+$ if $x_i$ is true in $R$ and, otherwise, $x_ix_i^-$. At this point, we have $|X|$ weight $1$ saturated edges.

We also saturate the weight $1$ edges of $\{v_iu_i \mid 1 \leq i \leq |X| \} \cup \{ c_i^+c_i^- \mid 1 \leq c_i \leq |C| \}$. Then, we will have $2|X|+|C|$ weight $1$ saturated edges in $M$ and, thus, $w(M) = 2|X|+|C|$.

Note that the matching is connected since, the vertices of $\{x_i, v_i \mid 1 \leq i \leq |X| \}$ induce a cycle in $G'(B)$ and are all saturated. The remaining edges that are not connected to those vertices are exactly the ones in $\{ c_i^+c_i^- \mid 1\leq i \leq |C| \}$. For $1\leq i \leq |C|$, the vertex $c_i^+$ is connected to at least one saturated vertex $x_i^+$ or $x_i^-$, which corresponds to a literal that resolves to true. This is due to the fact that we saturated the literal vertices according to the assignment $R$, that resolves $B$ to true. Thus, $M$ is connected.
\end{proof}

\begin{lemma}\label{lemma:wcm-planar-nn-volta}
Given the graph $G'(B)$ obtained by the formula $B$, and a connected matching $M$, $w(M)=2|X|+|C|$ in $G'(B)$, we can obtain in linear time a variable assignment of $X$ that resolves $B$ to true.
\end{lemma}
\begin{proof}
Denote $W_0$ and $W_1$ as the edge sets of $G'(B)$ whose weights are, respectively, $0$ and $1$. First, we show that a matching having weight $2|X|+|C|$ contains exactly $2|X|+|C|$ edges from $W_1$ and no edges from $W_0$.

Note that there can be at most $2|X|+|C|$ edges from $W_1$. This holds because, for each variable $x_i \in X$, there is at most one saturated edge of $\{x_i^+x_i, x_i^-x_i\}$ since both have an endpoint in vertex $x_i$. The other edges from $W_1$, $\{ c_i^+c_i^- \mid 1 \leq i \leq |C| \} \cup \{ v_iu_i \mid 1 \leq i \leq |X| \}$ can all be saturated.

Observe that each edge from $W_0$ has an endpoint in $S_1 = \{ c_i^+ \mid 1 \leq i \leq |C| \}$ or in $S_2 = \{ v_i \mid 1 \leq i \leq |X| \}$. Saturating any of these vertices by a $W_0$ edge will decrease the maximum number of saturated edges of $W_1$ and, thus, decrease the weight of $M$ such that $w(M) < k$. For instance, if we saturate $c_i^+ \in S_1$ or $v_i \in S_2$, we will not be able to saturate, respectively, $c_i^+c_i^-$ or $v_iu_i$.

Therefore, $|M \cap W_1| = 2|X|+|C|$ and $|M \cap W_0| =0$.

Moreover $M$ is connected, then, for each saturated edge $c_i^+c_i^-$, there is a saturated adjacent vertex, either $x_j^+$ or $x_j^-$, $x_j \in X$, $x_j \in c_i$. Those vertices are exactly the ones representing literals in the clause $c_i$ that resolves $c_i$ to true.

Also, the vertices $\{x_iv_i \in 1\leq i \leq |X|\}$ are all saturated and connected.

Therefore, to obtain $R$, for each variable $x_i \in X$, we can set $x_i = T$ if and only if $x_i^+$ is saturated, which can be done in linear time.
\end{proof}

\begin{theorem}\label{theo:wcm-planar-nn}
\pname{Weighted Connected Matching} is {\NPc} even for planar bipartite graphs whose edge weights are in $\{ 0,1 \}$
\end{theorem}
\begin{proof}
Proposition~\ref{prop:wcm-np} shows that the problem is in {\NP}. In Lemmas~\ref{lemma:wcm-planar-nn-ida} and \ref{lemma:wcm-planar-nn-volta}, we show transformations between solutions of \pname{Weighted Connected Matching} and \pname{Planar Monotone SAT}, which is {\NPc}. Thus, \pname{Weighted Connected Matching} can be reduced to \pname{Planar Monotone SAT} using a planar and bipartite graph whose edge weights are either $0$ or $1$. Therefore, \pname{Weighted Connected Matching} is {\NPc} even for bipartite planar graphs whose weights are in $\{ 0, 1\}$.
\end{proof}
\subsection{Example}

Let $B$ be an input for \pname{Planar Monotone SAT} defined as $(x_1 \vee x_2 \vee x_5) \wedge (x_2 \vee x_3 \vee x_4) \wedge (\overline{x_2} \vee \overline{x_4} \vee \overline{x_5})$.

Note that this formula is monotone as, in each clause, literals are all either positive or negative. Moreover, $G(B)$ is planar, and all edges representing positive(negative) literals are connected at the top(bottom) of the variable vertices. Observe an example of this graph in Figure~\ref{fig:wcm-planar-nn-gb}.

Thus, $B$ is in fact an input example of \pname{Planar Monotone SAT}, which can be used to generate the input $G'(B)$ and $k=2|X|+|C|=13$ for \pname{Weighted Connected Matching}.

Observe in Figure~\ref{fig:wcm-planar-nn} a connected matching with weight $13$ in $G'(B)$. Dashed and solid edges represent weights $0$ and $1$, respectively. This matching corresponds to the assignment $(T,F,T,T,T)$ of the variables $(x_1, x_2, x_3, x_4, x_5)$, in this order.

\begin{figure}
    \centering
    \includegraphics[width=0.8\textwidth]{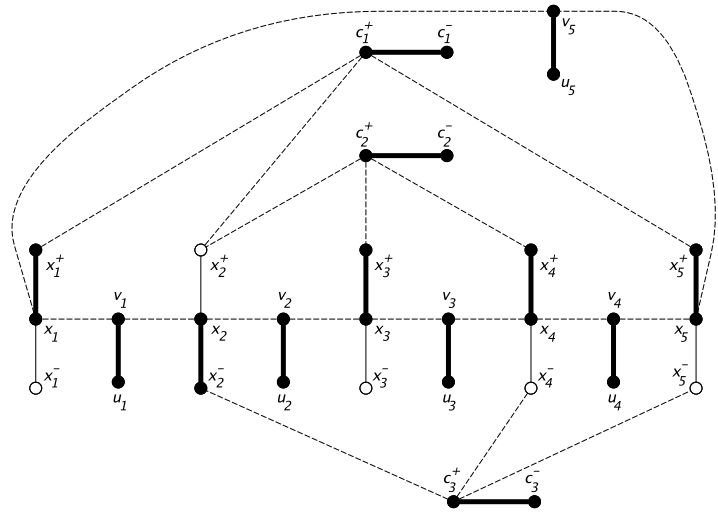}
    \caption{Example of a graph $G'(B)$ and a connected matching}
    \label{fig:wcm-planar-nn}
\end{figure}
\subsection{Chordal graphs}

We know that the version of \pname{Weighted Connected Matching} in which is allowed to have negative weight edges is {\NPc} for chordal graphs. This is due to the {\NPcness} for starlike graphs stated in Theorem~\ref{teo:wcm-starlike}.

Unlikely, we show that if restrict the weights to be non negative only, we can solve the problem in polynomial time for chordal graphs. For this purpose, we use a polynomial reduction to the problem \pname{Maximum Weight Perfect Matching}, which is in {\P}.


This section is organized as follows. First, we analyze matchings in general for chordal graphs, showing in Proposition~\ref{prop:wcm-choral-sep} that it is easy to obtain maximum weight matchings in which every separator has at most one of its vertices non saturated. This implies that the problem for chordal graphs without articulations can be solved with the same complexity as \pname{Weighted Matching}. Next, we show in Lemma~\ref{lem:wcm-chordal-art} that every graph has a maximum weight connected matching that saturates all its articulations. Finally, we use this property and build a polynomial reduction from \pname{Weighted Connected Matching} to \pname{Weighted Perfect Matching} for chordal graphs in general.

Let $G$ be a chordal graph and $M$ a maximum weight matching of $G$. By definition, we know that every minimal separator of $G$ is a clique. Then, observe that, if there is any minimal separator $S$ containing two non saturated vertices $v_1,v_2 \in S$, we can add $v_1v_2$ to $M$. After this, we should have a matching in which every minimal separator has at most one non vertex saturated, which implies Proposition~\ref{prop:wcm-choral-sep}.

\begin{proposition}\label{prop:wcm-choral-sep}
In a chordal graph, there is a maximum weight matching $M$ such that, for every minimal separator $S$, it holds that $|S-1| \leq |V(M) \cap S| \leq |S|$.
\end{proposition}

For this reason, given a maximum weight matching $M$, we can saturate a maximal set of edges having its endpoints in two non saturated edges of a minimal separator, obtaining a maximum weight connected matching.

Unlikely, when considering the problem for chordal graphs in general, whose minimal separators may be articulations, we can not rely on this solution.

Instead, we  use a polynomial reduction to \pname{Maximum Weight Perfect Matching}, which is in {\P}. This was first shown by Jack Edmonds~\cite{edmonds_maximum}, and, over time, better algorithms were given~\cite{mwpm}. In this problem, we are given a graph $G$ and we want to find a perfect matching $M$ whose sum of the edge weights is maximum. Note that input graphs must have perfect matchings.






In the next lemma, we show that there is always a maximum connected weighted matching that saturates all articulations of a graph.

\begin{lemma}\label{lem:wcm-chordal-art}
Let $G$ be a connected graph whose edge weights are non negative. There is a maximum weight connected matching $M$ that saturates all articulations.
\end{lemma}
\begin{proof}
Suppose that $v$ is an articulation that is not saturated by $M$. Also, denote $C = \{ C_1, C_2,\ldots, C_s \}$ by the connected components of $S = G - v$. We know that, by definition, $|C| \geq 2$. Note that the edges of $M$ have to be contained in exactly one component $C_i$ of $S$, because, otherwise, $G[M]$ would be disconnected.

We show that it is possible to increment the size of $M$ such that the weight is preserved or increased. Consider the component $C_j \in S$, $i \neq j$. Let $U$ be $V(C_i) \cap N(v)$ and $W$ be $V(C_j) \cap N(v)$. We consider two cases.

For the first case, there exists a vertex $u \in U$ not saturated. Observe that $v$ is adjacent to some saturated vertex, because, otherwise, it would be possible to add to $M$, alternatively, the edges of a path starting on $v$ and ending on a saturated vertex. Therefore, the matching $M \cup uv$ would still be connected.

For the second case, all vertices of $U$ are saturated. Still, it is possible to add to $M$ the edge $vw$, $w \in W$.

In both cases, $M$ can be incremented and its weight is not decreased, which implies that there is a maximum weight connected matching such that $v$ is saturated.
\end{proof}

Without loss of generality, by Lemma~\ref{lem:wcm-chordal-art}, note that a maximum weight matching in a chordal graph saturates all articulations. This fact will be used later in our proof.

For the reduction, from a chordal graph $G=(V,E)$, we build an input graph $G_p=(V_p,E_p)$ for the \pname{Maximum Weight Perfect Matching}, built by the following rules.

\begin{enumerate}[(I)]
    \item Set $V_p = V$. Also, if $|V|$ is odd, add a vertex $h$ to $V_p$.
    \item Set $E_p = E$. Moreover, for each pair of vertices $v_1,v_2 \in V$, if $v_1v_2 \notin E$, add the edge $v_1v_2$ having weight $0$ to $E_p$.
\end{enumerate}

In Lemma~\ref{lem:wcm-chordal-perfect}, we show that this graph can indeed be used for the \pname{Maximum Weight Perfect Matching} input, as it has a perfect matching. In the following, Lemma~\ref{lem:wcm-chordal-volta} shows that the answer of this instance can be used to obtain a maximum weight connected matching in linear time. Finally, Theorem~\ref{teo:wcm-chordal-poly} concludes the complexity of \pname{Maximum Weight Connected Matching} for chordal graphs.

\begin{lemma}\label{lem:wcm-chordal-perfect}
The graph $G_p$ has a perfect matching.
\end{lemma}
\begin{proof}
Let $T$ be a block-cutpoint tree of $G$ rooted in a vertex denoted by $r$. Also, denote $A$ as the set of articulations of $G$. Next, we iteratively build a perfect matching $M$ in $G_p$. 

Starting from $r$, we traverse $T^r$ in preorder and, when visiting a non saturated vertex $v$ that is contained in $A$, we saturate $v$ by the following way. If $v$ has a non leaf child $c$, we saturate $v$ with any child of $c$ contained in $A$. Otherwise, let $c$ be one of $v$'s children. Then, we saturate $v$ with any vertex $u$ contained in the block of $G$ represented by $c$.

As an example, in Figure~\ref{fig:wcm-chordal-perfect-algo-fig}, consider the graph and its block-cutpoint tree rooted in $B_4$. The saturated edges by this procedure that are incident to articulations, in order, can be $c_3c_2$, $c_1v$, $c_5c_7$ and $c_4c_6$.

\begin{figure}
    \centering
    \subfloat[]{\includegraphics[width=0.5\textwidth]{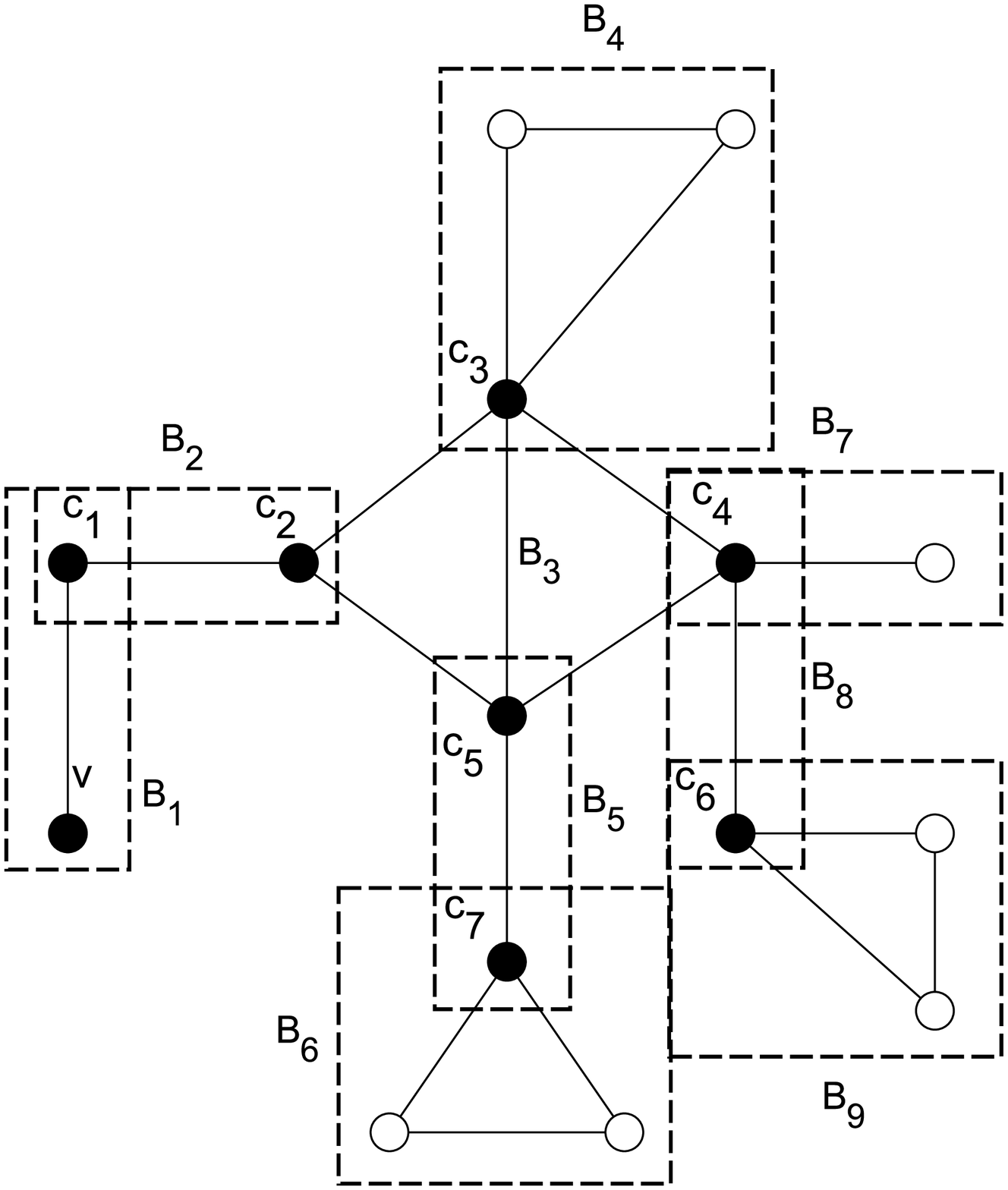}\label{fig:wcm-chordal-perfect-algo}}
    \hspace{0.1\textwidth}
    \subfloat[]{\includegraphics[width=0.3\textwidth]{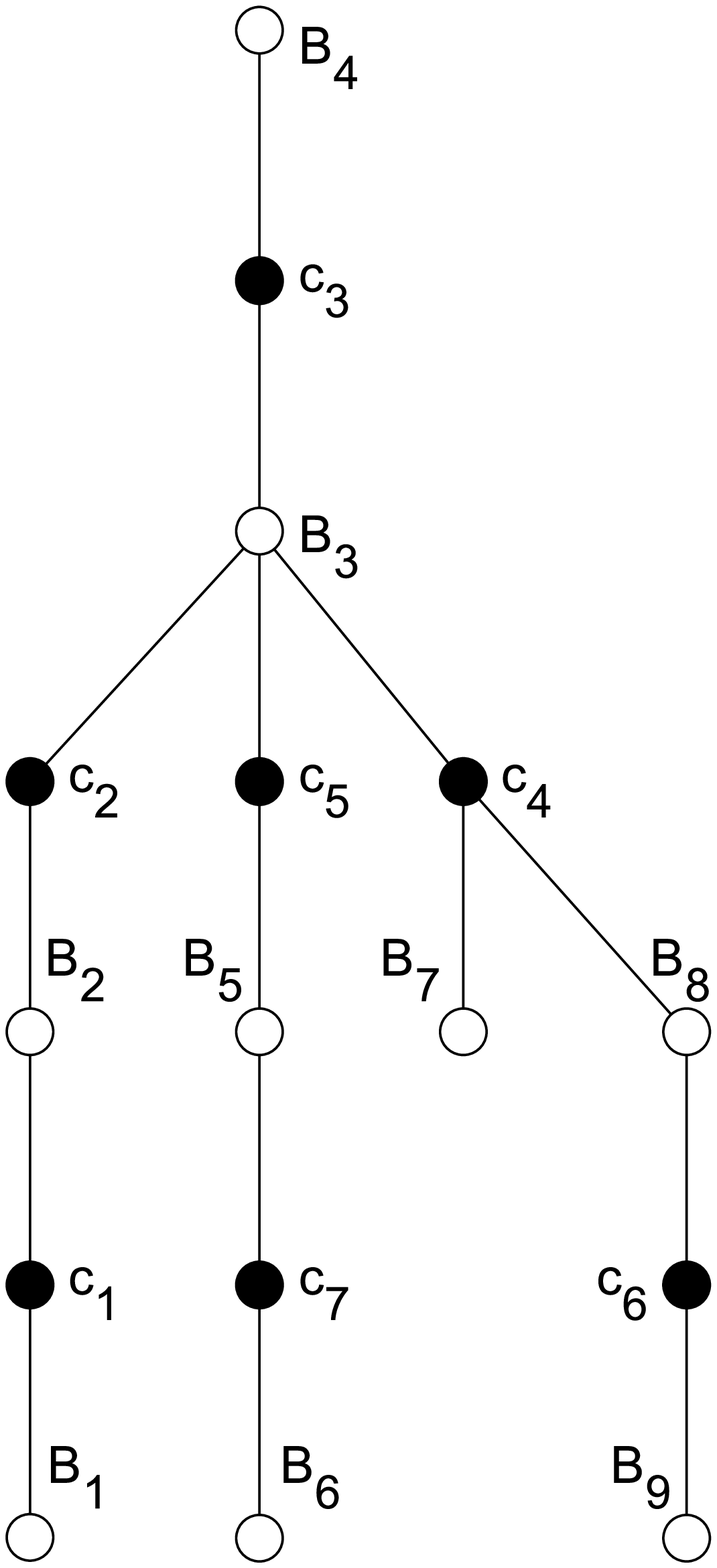}\label{fig:wcm-chordal-perfect-algo-cutpoint}}
    
    \caption{A chordal graph in (a) and its block-cutpoint tree in (b).}%
    \label{fig:wcm-chordal-perfect-algo-fig}
\end{figure}

Note that, when visiting a non saturated vertex $v \in A$, none of its descendant vertices are saturated. Moreover, $v$ is not a leaf, because, as it is an articulation, it is contained in at least two blocks of $G$. At least one of these blocks is a child of $v$ in $T^r$. Thus, there will always be at least one vertex to saturate with $v$.

Finally, we saturate all the remaining vertices of $V_p$, which induce a clique, by any set of disjoint edges of $E_p$. Observe that this procedure saturates all vertices and, therefore, $G_p$ has a perfect matching.
\end{proof}

\begin{lemma}\label{lem:wcm-chordal-volta}
Given a maximum weight perfect matching $M_p$ in $G_p$, we can obtain, in linear time, a maximum weight connected matching $M$ of $G$.
\end{lemma}
\begin{proof}
We build a maximum weight connected matching $M$ in $G$.

First, (i) we set $M = M_p \cap E(G)$. At this point, there may be vertices not saturated and connected by weight $0$ edges. Finally, (ii) we saturate a maximal set of those vertices. Note that $w(M) = w(M_p)$ and this can be done in linear time.

Now, we prove that $M$ is connected because, in each separator $S$ of $G$, there is at least one saturated vertex. Let's suppose this is not true and, then, let $v_1,v_2 \in S$ be two vertices not saturated. As they are not saturated in $M$, its incident saturated edges in $M_p$ are in $E(G_p) \setminus E(G)$, whose weights are zero. Moreover, $w(v_1v_2) > 0$, since both $v_1$ and $v_2$ were not saturated in (ii). Let $v_1u_1$ and $v_2u_2$ be the edges that saturate the vertices $v_1$ and $v_2$ in $M_p$. If we change the edges $\{v_1u_1, v_2u_2\}$ by $\{v_1v_2, u_1u_2\}$, the matching would still be perfect, but its weight would be greater than $M_p$. This is a contradiction and, then, $M_p$ can not be a maximum weight perfect matching.

Therefore, for each separator $S$ of $G$, there is at least one saturated vertex. Also, note that all articulations of $G$ are saturated by $M$, since all edges incident to articulations are in $E(G_p) \cap E(G)$. For this reason, each separator has at least one saturated vertex by $M$ and, thus, $M$ is connected.

Now, we show that $M$ is maximum. Suppose there is a connected matching $M'$ with greater weight than $w(M)$. Moreover, based on Lemma~\ref{lem:wcm-chordal-art}, we assume that $M'$ saturates all articulations.

Now, we build a perfect matching $M_p'$ from $M'$. First, set $M_p' = M'$. Since all articulations of $G$ are saturated by $M'$, then, all the remaining vertices induce a clique. Finally, we saturate all those vertices by any disjoint set of edges.

Observe that $w(M_p') \geq w(M_p) > w(M_p) = w(M)$. This is a contradiction, because $M_p$ is a maximum weight perfect matching in $G_p$. Therefore the matching $M$ is a maximum weight connected matching in $G$ and it can be obtained in linear time.
\end{proof}

\begin{theorem}\label{teo:wcm-chordal-poly}
\pname{Maximum Weight Connected Matching} for chordal graphs whose edge weights are all non negative is in {\P} and its complexity is the same as \pname{Maximum Weight Perfect Matching}.
\end{theorem}
\begin{proof}
To solve \pname{Maximum Weight Connected Matching} for a input graph $G$, we can build, in linear time, the graph $G_p$ and use it as an input of a \pname{Maximum Weight Perfect Matching} algorithm. Given this answer, we can obtain a maximum weight connected matching of $G$ in linear time, as stated in Lemma~\ref{lem:wcm-chordal-volta}. Therefore, we can solve \pname{Maximum Weight Connected Matching} with the same complexity as \pname{Maximum Weight Perfect Matching}. 
\end{proof}

\subsubsection{Example}

Let the input graph $G$ of \pname{Maximum Weight Connected Matching} be the one illustrated in Figure~\ref{fig:wcm-chordal-poly-mm}, where we show a maximum weight matching, whose weight is $44$. Observe that this matching is not connected as $f$ is not saturated.

In Figure~\ref{wcm-chordal-poly-mwpm}, we show the subgraph induced of $G_p$ by a maximum weight perfect matching, whose weight is $42$. This matching can be converted to a maximum weight connected matching in $G$ having the same weight by removing the edge $lf$. 

\begin{figure}
    \centering
    \subfloat[]{\includegraphics[width=0.5\textwidth]{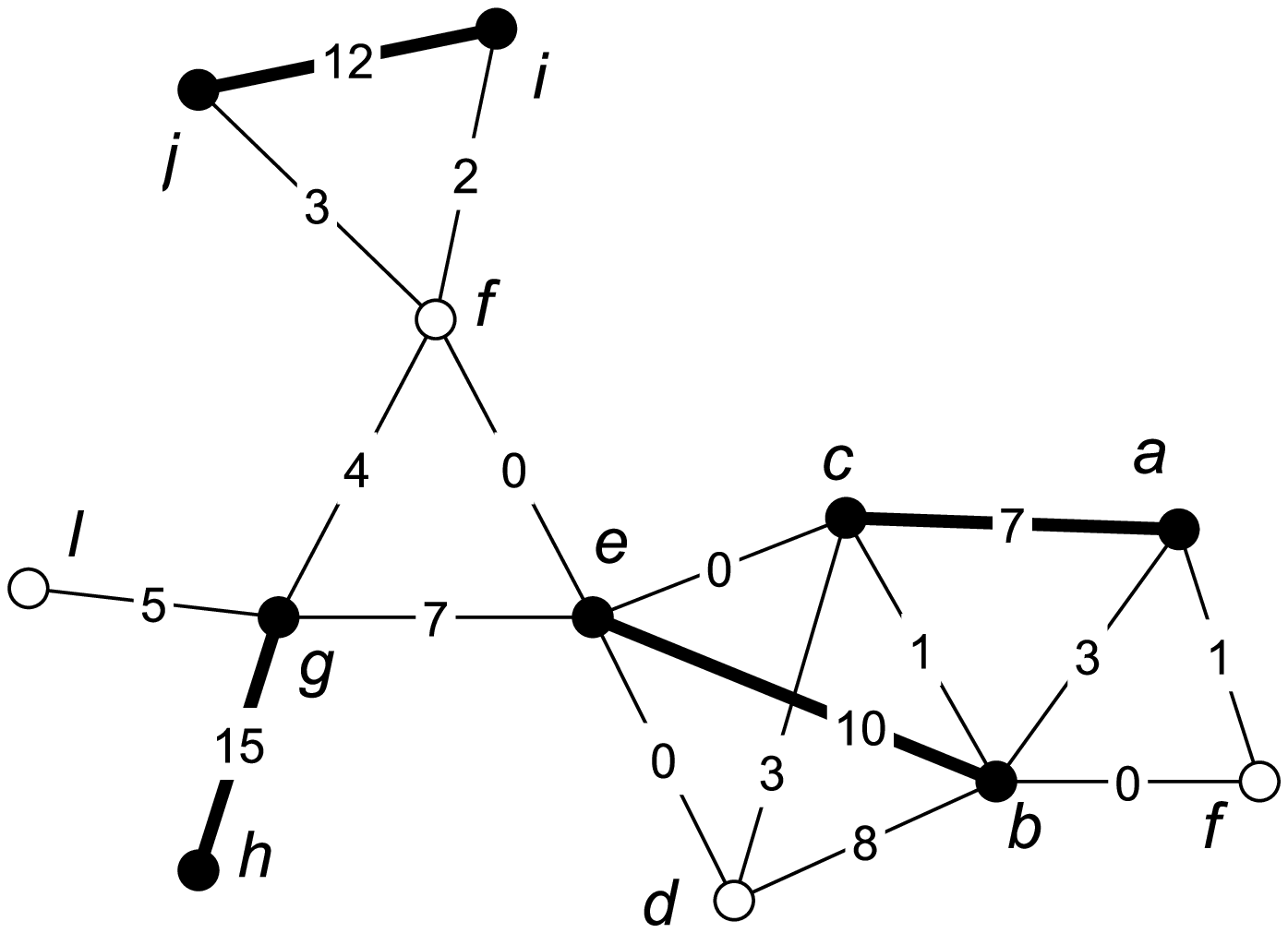}\label{fig:wcm-chordal-poly-mm}}
    
    \subfloat[]{\includegraphics[width=0.50\textwidth]{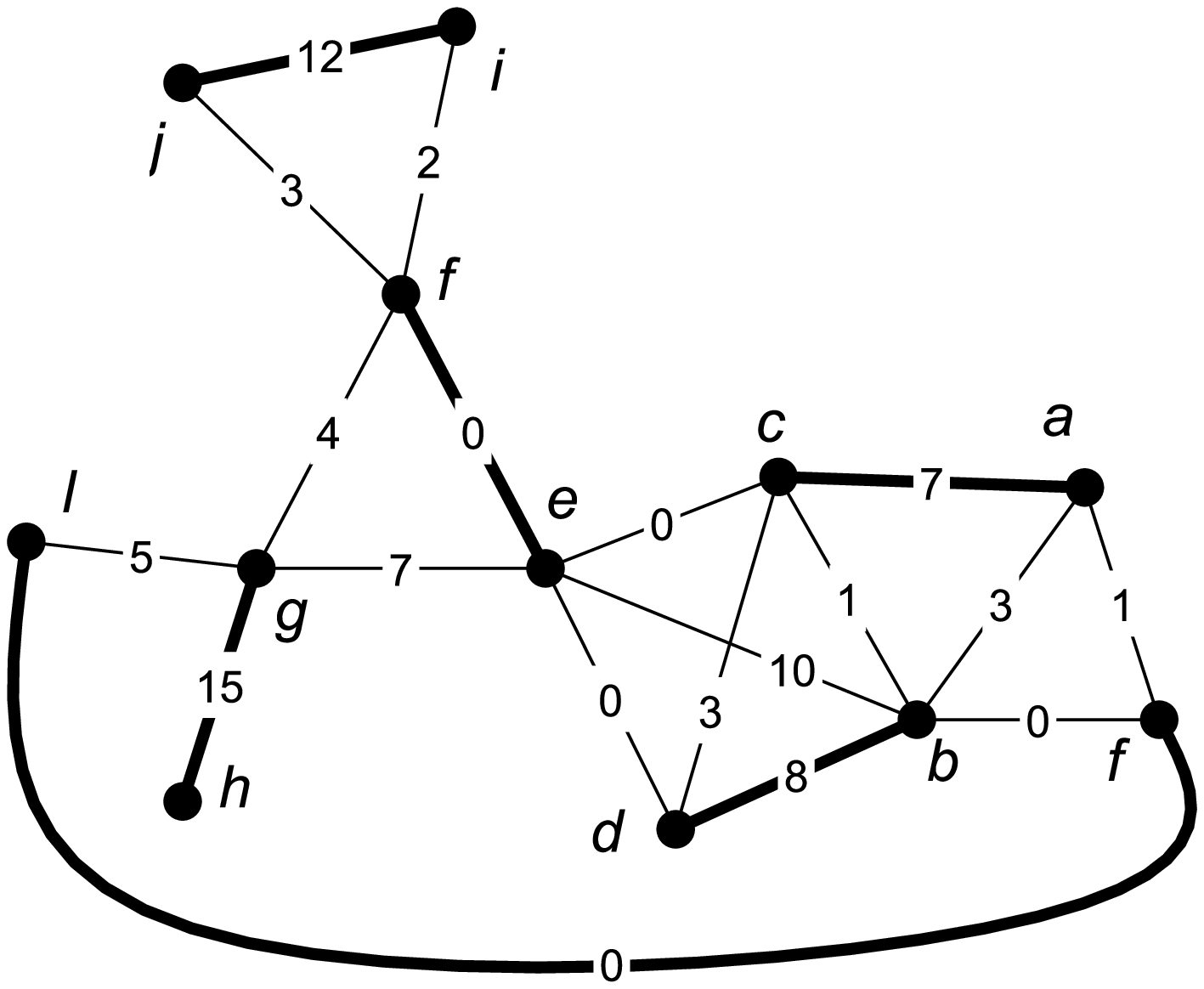}\label{wcm-chordal-poly-mwpm}}
    \subfloat[]{\includegraphics[width=0.50\textwidth]{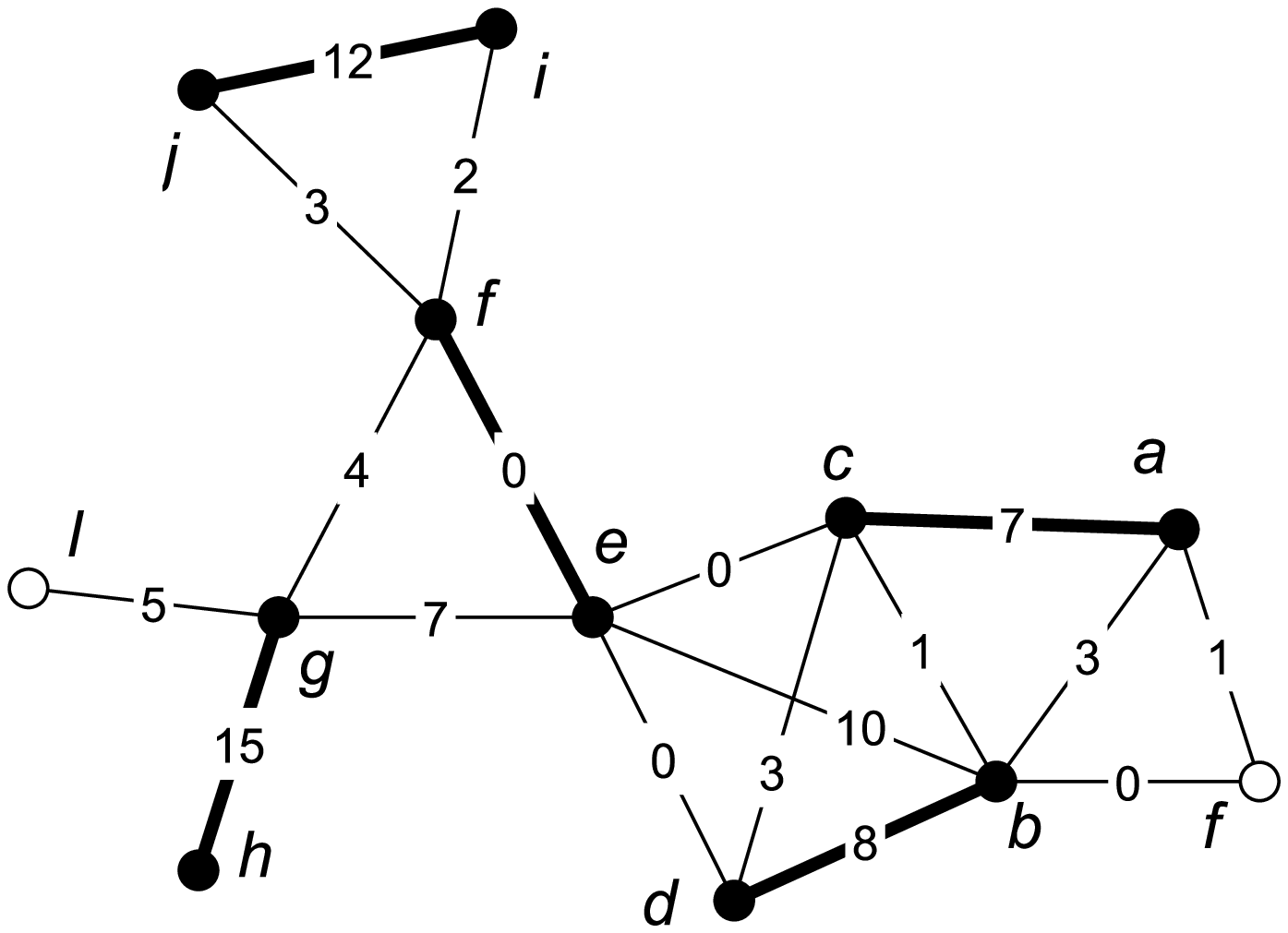}\label{fig:wcm-chordal-poly-mcm}}
    
    \caption{A chordal graph $G$ and its maximum weight matching in (a), a maximum weight perfect matching of $G_p$ in (b) and its corresponding maximum weight connected matching of $G$ in (c).}%
    \label{fig:wcm-chordal-poly-transformation}
\end{figure}


\section{Kernelization}\label{sec:kernel}

In this section, we present our kernelization results. In particular, we show that \pname{Weighted Connected Matching} on parameterized by vertex cover does not admit a polynomial kernel, unless $\NP \subseteq \coNP/\poly$, even if the input is restricted to bipartite graphs of bounded diameter and the allowed weights are contained in $\{0, 1\}$.

We prove our result through an OR-cross-composition~\cite{cross_composition} from the \pname{3SAT} problem, and are heavily inspired by the proof described in Section~\ref{sec:bipartite_complexity}.
Formally let, $\mathcal{H} = \{(X_1, \mathcal{C}_1), \dots, (X_t, \mathcal{C}_t)\}$ be the $t$ \pname{3SAT} instances such that $X_i = X = \{x_1, \dots, x_n\}$ for every $i \in [t]$.
Also, let $\mathcal{C} = \bigcup_{i \in [t]} \mathcal{C}_i$.
Finally, let $(G, k)$ be the \pname{Weighted Connected Matching} instance we are going to build.
We begin our construction by adding to $G$ a pair of vertices $c_j,c_j'$ for each $C_j \in \mathcal{C}$ and a unit weight edge between them.
Then, for each $x_i \in X$, we add vertices $x_i^-, x_i^*, x_i^+$ and edges $x_i^-x_i^*, x_i^*x_i^+$, each of weight 1. Now, for each $C_j \in \mathcal{C}$ and $i \in [n]$, if $x_i \in C_j$, we add the 0-weight edge $x_i^+c_j$ to $G$, otherwise, if $\overline{x_i} \in C_j$, we add the weight 0 edge $x_i^-c_j$.
We conclude this first part of the construction by adding a pair of vertices $h,h'$ to $G$, making them adjacent with an edge of weight 1, and adding an edge of weight 0 between $h$ and $x_i$ for every $x_i \in X$.
At this point, we have an extremely similar graph to the one constructed in Section~\ref{sec:bipartite_complexity}.

For the next part of the construction, we add a copy of $K_{1,t}$, where the vertex on the smaller side is labeled $q$ and, the vertices on the other side are each assigned a unique label from $Y = \{y_1, \dots, y_t\}$.
Now, for each $y_\ell \in Y$ and $C_j \in \mathcal{C} \setminus \mathcal{C}_\ell$, we add the 0-weight edges $c_j'y_\ell$ and $hy_\ell$.
Finally, we set $k = |\mathcal{C}| + |\mathcal{X}| + 2$.
Note that $|V(G)| = 3|X| + 2|C| + |Y| + 3 \leq 3|X| + 2|X|^3 |Y| + 3$, which implies that $V(G) \setminus Y$ is a vertex cover of $G$ of size $\bigO{|X|^3}$, as required by the cross-composition framework.
Moreover, note that $G$ is bipartite, as we can partition it as follows: $L = {q, h} \cup \{x_i^+, x_i^- \mid i \in [n]\} \cup \{c_j' \mid C_j \in \mathcal{C}\}$ and $R = V(G) \setminus L$, where both $L$ and $R$ are independent sets.

\begin{lemma}
    \label{lem:no_kernel_vc_forward}
    If $(X, \mathcal{C}_\ell)$ admits a solution, then $(G, k)$ also admits a solution.
\end{lemma}

\begin{proof}
    Let $\varphi$ be a satisfying assignment for $(X, \mathcal{C}_\ell)$.
    We build the solution M for $(G, k)$ as follows.
    First, for each $x_i \in X$, if $\varphi(x_i) = 0$, add edge $x_i^-x_i^*$ to $M$, otherwise add edge $x_i^+x_i^*$ to $M$, with a total weight of $|X|$ after this step.
    Now, for each $C_j \in \mathcal{C}$, and $c_jc_j'$ to $M$, reaching $|\mathcal{C}| + |\mathcal{C}| $ weight.
    Finally, add $hh'$ and $qy_\ell$ to $M$, so now $M$ has $ |X| + |\mathcal{C}| + 2 = k$ weight.
    Note that $M$ is a matching.
    To see that it induces a connected graph, first observe that $\{q, y_\ell, h, h', x_1^*, \dots, x_,^*\} \cup \bigcup_{i \in [n]} \{x_i^-, x_i^+\} \cap M$ are all part of the same connected component $Q$.
    Moreover, for every $C_j \notin \mathcal{C}_\ell$, we have both $c_j'$ and $c_j$ also belong to $Q$ since $c_j'y_\ell \in E(G)$.
    For each $C_j \in \mathcal{C}_\ell$, suppose that $x_i \in C_j$ and $\varphi(x_i) = 1$; note that $x_i^+ \in Q$, so it holds that $c_j$ and $c_j'$ are also in $Q$, completing the proof.
\end{proof}

\begin{lemma}
    If $(G,k)$ admits a connected matching $M$ of weight at least $k$, then there is some $(X, \mathcal{C}_\ell) \in \mathcal{H}$ that admits a solution.
\end{lemma}

\begin{proof}
    First, note that $k$ is also the weight of a maximum weighted matching of $G$, which is achieved by picking all edges $c_jc_j'$, edge $hh'$, one edge incident to $q$, and one edge of weight one incident to each $x_i^*$.
    As such, we observe that there is one edge $qy_\ell \in M$ and, furthermore, no other $y_p \in M$, otherwise they would either be matched to $h$ or to some $c_j'$; in either case we would have $w(M) < k$, since we would be replacing an edge of weight one with one of weight zero, and neither $h'$ nor $c_j$ can be matched with other edges of larger weight.
    Moreover, this implies that each $x_i^*$ is matched to either $x_i^+$ or $x_i^-$, otherwise we would also not be able to achieve the necessary weight.
    As such, for each $x_i$, we set $\varphi(x_i) = 1$ if and only if $x_i^+$ is matched to $x_i^*$.
    Finally, note that, for each $C_j \in C_\ell$, there must be a path between $c_j'$ and $q$ passing through some $x_i^*$ and, furthermore, this path must pass through either $x_i^+$ if $x_i \in C_j$ or $x_i^-$ if $\overline{x_i} \in C_j$.
    This, in turn, implies that there is a literal of $C_j$ that evaluates to true and satisfies $C_j$.
    Consequently, every $C_j$ is satisfied and $\varphi$ is a solution to $(X, \mathcal{C}_\ell)$.
\end{proof}

Combining the two previous Lemmas and our observations and the end of the construction of $(G,k)$, we immediately obtain the our theorem.

\begin{theorem}
    \label{thm:no_kernel_vc}
    Unless $\NP \subseteq \coNP/\poly$, \pname{Weighted Connected Matching} does not admit a polynomial kernel when parameterized by vertex cover and required weight even if the input graph is bipartite and the edge weights are in $\{0, 1\}$.
\end{theorem}

\section{Single exponential time algorithms}
\label{sec:single_exp}

The result in this section relies on the \textit{rank based approach} of Bodlaender et al.~\cite{lattice_algebra}, which requires the additional definitions we give below.
Let $U$ be a finite set, $\Pi(U)$ denote the set of all partitions of $U$, and $\sqsubseteq$ be the coarsening relation defined on $\Pi(U)$, i.e given two partitions $p,q \in \Pi(U)$, $p \sqsubseteq q$ if and only if each block of $q$ is contained in some block of $p$.
It is known that $\Pi(U)$ together with $\sqsubseteq$ form a lattice, upon which we can define the usual \textit{join} operator $\join$ and \textit{meet} operator $\meet$~\cite{lattice_algebra}.
The join operation $p \join q$ works as follows: let $H$ be the graph where $V(H) = U$ and $E(H) = \{uv \mid \{\{u,v\}\} \sqsubseteq p \vee \{\{u,v\}\} \sqsubseteq q\}$; $S \subseteq U$ is block of $p \join q$ if and only if $S$ induces a maximal connected component of $H$.
The result of the meet operation $p \meet q$ is the unique partition such that each block is formed by the non-empty intersection between a block of $p$ and a block of $q$.
Given a subset $X \subseteq U$ and $p \in \Pi(U)$, $p_{\downarrow X} \in \Pi(X)$ is the partition obtained by removing all elements of $U \setminus X$ from $p$, while, for $Y \supseteq U$, $p_{\uparrow Y} \in \Pi(Y)$ is the partition obtained by adding to $p$ a singleton block for each element in $Y \setminus U$.
For $X \subseteq U$, we shorthand by $U[X]$ the partition where one block is precisely $\{X\}$ and all other are the singletons of $U \setminus X$; if $X = \{a,b\}$, we use $U[ab]$.

A set of \textit{weighted partitions} of a ground set $U$ is defined as $\mathcal{A} \subseteq \Pi(U) \times \mathbb{N}$.
To speed up dynamic programming algorithms for connectivity problems, the idea is to only store a subset $\mathcal{A}' \subseteq \mathcal{A}$ that preserves the existence of at least one optimal solution.
Formally, for each possible extension $q \in \Pi(U)$  of the current partitions of $\mathcal{A}$ to a valid solution, the optimum of $\mathcal{A}$ relative to $q$ is denoted by $\opt(q, \mathcal{A}) = \min \{w \mid (p, w) \in \mathcal{A}, p \join q = \{U\}\}$.
$\mathcal{A}'$ \textit{represents} $\mathcal{A}$ if $\opt(q, \mathcal{A}') = \opt(q, \mathcal{A})$ for all $q \in \Pi(U)$.
The key result of~\cite{lattice_algebra} is given by Theorem~\ref{thm:reduce}.

\begin{theorem}[3.7 of~\cite{lattice_algebra}]
    \label{thm:reduce}
    There exists an algorithm that, given $\mathcal{A}$ and $U$, computes $\mathcal{A}'$ in time $|\mathcal{A}|2^{(\omega-1)|U|}|U|^\bigO{1}$ and $|\mathcal{A}'| \leq 2^{|U|-1}$, where $\omega$ is the matrix multiplication constant.
\end{theorem}

A function $f : 2^{\Pi(U) \times \mathbb{N}} \times Z \mapsto 2^{\Pi(U) \times \mathbb{N}}$ is said to \textit{preserve representation} if $f(\mathcal{A}', z) = f(\mathcal{A}, z)$ for every $\mathcal{A}, \mathcal{A}' \in \Pi(U) \times \mathbb{N}$ and $z \in Z$; thus, if one can describe a dynamic programming algorithm that uses only transition functions that preserve representation, it is possible to obtain $\mathcal{A}'$.
In the following lemma, let $\rmc{\mathcal{A}} = \{(p,w) \in \mathcal{A} \mid \nexists (p, w') \in \mathcal{A}, w' < w\}$.

\begin{lemma}[Proposition 3.3 and Lemma 3.6 of~\cite{lattice_algebra}]
    \label{lem:functions}
    Let $U$ be a finite set and $\mathcal{A} \subseteq \Pi(U) \times \mathbb{N}$.
    The following functions preserve representation and can be computed in $|\mathcal{A}|\cdot|\mathcal{B}|\cdot|U|^\bigO{1}$ time.
    
    \begin{itemize}
        \item[\textbf{Union.}] For $\mathcal{B} \in \Pi(U) \times \mathbb{N}$, $\mathcal{A} \union \mathcal{B} = \rmc{\mathcal{A} \cup \mathcal{B}}$.
        \item[\textbf{Insert.}] For $X \cap U = \emptyset$, $\ins(X, \mathcal{A}) = \{(p_{\uparrow X \cup U}, w) \mid (p,w) \in \mathcal{A}\}$.
        \item[\textbf{Shift.}] For any integer $w'$, $\shift(w', \mathcal{A}) = \{(p,w + w') \mid (p,w) \in \mathcal{A}\}$.
        \item[\textbf{Glue.}] Let $\hat{U} = U \cup X$, then $\glue(X, \mathcal{A}) = \rmc{\left\{(\hat{U}[X] \join p_{\uparrow \hat{U}}, w) \mid (p,w) \in \mathcal{A}\right\}}$.
        \item[\textbf{Project.}] $\proj(X, \mathcal{A}) = \rmc{\{(p_{\downarrow \overline{X}}, w) \mid (p,w) \in \mathcal{A}, \forall u \in X : \exists v \in \overline{X} : p \sqsubseteq U[uv]\}}$, if $X \subseteq U$.
        \item[\textbf{Join.}] If $\hat{U} = U \cup U'$, $\mathcal{A} \subseteq \Pi(U) \times \mathbb{N}$ and $\mathcal{B} \in \Pi(U') \times \mathbb{N}$, then $\joinf(\mathcal{A}, \mathcal{B}) = \rmc{\{(p_{\uparrow\hat{U}} \join q_{\uparrow \hat{U}}, w + w') \mid (p,w) \in \mathcal{A}, (q, w') \in \mathcal{B}\}}$.
    \end{itemize}
\end{lemma}

A \textit{tree decomposition} of a graph $G$ is a pair $\td{T} = \left(T, \mathcal{B} = \{B_j \mid j \in V(T)\}\right)$, where $T$ is a tree and $\mathcal{B} \subseteq 2^{V(G)}$ is a family where: $\bigcup_{B_j \in \mathcal{B}} B_j = V(G)$;
for every edge $uv \in E(G)$ there is some~$B_j$ such that $\{u,v\} \subseteq B_j$;
for every $i,j,q \in V(T)$, if $q$ is in the path between $i$ and $j$ in $T$, then $B_i \cap B_j \subseteq B_q$.
Each $B_j \in \mathcal{B}$ is called a \emph{bag} of the tree decomposition.
$G$ has treewidth has most $t$ if it admits a tree decomposition such that no bag has more than $t+1$ vertices.
For further properties of treewidth, we refer to~\citep{treewidth}.
After rooting $T$, $G_x$ denotes the subgraph of $G$ induced by the vertices contained in any bag that belongs to the subtree of $T$ rooted at bag $x$.
An algorithmically useful property of tree decompositions is the existence of a \emph{nice tree decomposition} that does not increase the treewidth of $G$.

\begin{definition}[Nice tree decomposition]
    \label{def:nice_tree}
    A tree decomposition $\td{T}$ of $G$ is said to be \emph{nice} if its tree is rooted at, say, the empty bag $r(T)$ and each of its bags is from one of the following four types:
    \begin{enumerate}
        \item \emph{Leaf node}: a leaf $x$ of T with $B_x = \emptyset$.
        \item \emph{Introduce vertex node}: an inner bag $x$ of $\td{T}$ with one child $y$ such that $B_x \setminus B_y = \{u\}$.
        \item \emph{Forget node}: an inner bag $x$ of $\td{T}$ with one child $y$ such that $B_y \setminus B_x = \{u\}$.
        \item \emph{Join node}: an inner bag $x$ of $\td{T}$ with two children $y,z$ such that $B_x = B_y = B_z$.
    \end{enumerate}
\end{definition}

\begin{theorem}
    There is an algorithm for \pname{Maximum Weight Connected Matching} that, given a nice tree decomposition of width $t$ of the $n$-vertex input graph $G$ rooted at the forget node for some terminal $r \in K$, runs in time $2^\bigO{t}n^\bigO{1}$.
\end{theorem}

\begin{proof}
    Let $G$ be the input graph to \pname{Maximum Weight Connected Matching}, $\rho : E(G) \mapsto \mathbb{R}$ be the weighting of the edges, and $\td{T}^* = (T^*, \mathcal{B}^*)$ be a tree decomposition of width $t$ of $G$; in a slight abuse of notation, for $L \subseteq E(G)$, we define $\rho(L) = \sum_{e \in L} \rho(e)$.
    We also assume that $G$ is connected since, if it is not connected, we can run an algorithm in parallel for each component, which have treewidth bounded by that of $G$, and output the maximum of all these distinct components.
    In the first step of our algorithm, we pick a vertex $\pi \in V(G)$ and create a tree decomposition $\td = (T, \mathcal{B})$ rooted at a node $r$ that corresponds to an empty bag that is also a forget bag for vertex $\pi$; for each such choice of $\pi$ we run the dynamic programming algorithm we describe in the remainder of this proof.
    For each node $x \in V(T)$, we compute the table $f_x(S, U) \subseteq \Pi(B_x) \times \mathbb{R}$, with $S \subseteq B_x$ and $U \subseteq B_x \setminus S$.
    If we have a weighted partition $(p, w) \in f_x(S,U)$, then we want to ensure that there is a (partial) solution $M_x$ with the following properties: (i) every vertex of $S$ is already matched in $G_x$, (ii) vertices of $U$ are half-matched, i.e. they have yet to be matched to other vertices but are used to determine connectivity of $G[M_x]$, and (iii) $\rho(M_x) = \sum_{e \in M_x} \rho(e) = w$.
    Note that vertices in $U$ must already be accounted for when determining the connected components induced by $M_x$.
    After every operation involving families of partitions, we apply the algorithm of Theorem~\ref{thm:reduce}.
    We divide our analysis in the four cases of Definition~\ref{def:nice_tree}, where $x$ is the bag for which we currently want to compute the dynamic programming table.
    
    \noindent \textbf{Leaf node.} Since $B_x = \emptyset$, the only possible connected matching of $G_x$ is the empty matching, so we define:

    \begin{equation*}
        \centering
        f_x(\emptyset, \emptyset) = \{(\emptyset, 0)\}
    \end{equation*}

    \noindent \textbf{Introduce node.} Let $y$ be the child of of $x$ and $\{v\} = B_x \setminus B_y$. We define our transition as follows, where $\mathcal{A}_x(S, U, u) = \ins(\{v\}, f_y(S \setminus \{u,v\}, U \cup \{u\}))$:
    
    \begin{equation*}
        \centering
        \hfill f_x(S, U) =
        \begin{cases}
            f_y(S, U), &\text{ if } v \notin S \cup U \text{;}\\
            \glue(N_S(u), \ins(\{v\}, f_y(S, U \setminus \{v\})), &\text{ if } v \in U \text{;}\\
            \bigunion_{u \in N_S(v)} \shift(\rho(uv), \glue\left(N_S(v), \mathcal{A}_x(S,U,u)\right)), &\text{ otherwise.}\\
        \end{cases}
    \end{equation*}
    
    If $v$ is not in $M_x$, then $M_x$ is also a partial solution to $G_y$ which, by induction, is represented by an element of $f_y(S, U)$, which is covered by the first case of the equation.
    On the other hand, if $v \in M_x$, then $v$ is either matched to a vertex in $S$, or it is a half-matched vertex.
    If it is half matched, then $M_y = M_x \setminus \{v\}$ must be a partial solution to $G_y$ and, furthermore, must be represented in $f_y(S, U \setminus \{v\})$, since its matched edges are the same as in $M_x$ and all other half-matched vertices of $M_x$ also exist in $M_y$; this situation is covered in the second case of the equation, where we must further coarsen the partition that represents $M_y$ by joining the blocks that have neighbors of $v$ in them.
    Finally, if $v \in S$ and $uv \in M_x$, then it must be the case that $u \in S$, since $v \notin V(G_y)$.
    As such, $M_y = \{u\} \cup M_x \setminus \{uv\}$ must be a partial solution to $G_y$ with $u$ being a half-matched vertex, i.e. $M_y$ must be represented in $f_y(S \setminus \{u,v\}, U \cup \{u\})$, which holds by induction.
    This final case is represented in the third case of the previous equation; note that we must add the weight of $uv$ to the weight of $M_y$ and join its connected components that contain vertices of $N_S(v)$.
    
    \noindent \textbf{Forget node.} Let $y$ be the child of $x$ and $\{v\} = B_y \setminus B_x$. We compute our table as follows:
    
    \begin{equation*}
        \centering
        f_x(S, U) = f_y(S, U) \union \proj(\{v\}, f_y(S \cup \{v\}, U))
    \end{equation*}
    
    First, consider the case where $v \notin M_x$, and note that $M_x$ is also a partial solution to $G_y$ and, consequently, must be represented by $f_y(S, U)$.
    On the other hand, if $v \in M_x$, then $v$ must be matched to some vertex of $M_x$, otherwise $M_x$ would not be extensible to a matching (i.e. without half-matched vertices) of $G$, since $v \notin G \setminus G_x$.
    These two cases are represented by the two right-hand-side terms of our previous equation.
    
    \noindent \textbf{Join node.} Finally, if $y,z$ are the children of $x$, then we compose our table according to the following equation:
    
    \begin{equation*}
        \centering
        \hfill f_x(S, U) = \bigunion_{Y \subseteq S} \joinf(f_y(Y, U \cup (S \setminus Y)), f_z(S \setminus Y, U \cup Y))
    \end{equation*}
    
    Where the union operator runs over all subsets of $S$.
    Let $M_x$ be a partial solution of $G_x$, $Y \subseteq S$ be the vertices in $S$ matched to vertices of $G_y$, and $M_y$ be the subset of $M_x$ restricted to $G_y$.
    Observe that $M_y \cap S = Y$ since the vertices of $Y$ are precisely those of $B_y$ matched in $M_y$.
    Moreover, $M_y \cap B_y \setminus Y = U \cup (S \setminus Y)$ are the vertices of $M_x$ not matched in $M_y$; they must, however, be half-matched vertices of $M_y$ since they are (half-)matched vertices of $M_x$.
    Consequently, $M_y$ is represented by a partition $(p, w_y) \in f_y(Y, U \cup (S \setminus Y))$.
    Now, let $M_z \subseteq M_x$ be the partial solution to $G_z$ where $S \setminus Y$ are the vertices of $S$ not matched by $M_y$.
    Note that $U \cup Y$ are precisely the vertices half-matched vertices of $M_z$ since they must be either half-matched in $M_x$ ($U$) or matched in $M_y$.
    As such, $M_z$ is represented by a partition $(q, w_z) \in f_z(S \setminus Y, U \cup Y)$.
    Finally, we have that $p \join q$ yields same partition of $S \cup U$ as $M_x$ since $M_y \cup M_z = M_x$, and, furthermore, $\rho(M_x) = \rho(M_y) + \rho(M_z) = w_y + w_z$ since no edge matched by $M_y$ is present in $M_z$ and vice-versa and $E(M_y) \cup E(M_z) = E(M_x)$.
    These are the exact properties given by the $\joinf$ operation; since the above equation runs over all subsets of $S$, $M_x$ will be represented by $f_x(S, U)$.
    
    In order to obtain the solution to $G$, first observe that, since $G$ is connected and the root $r$ of $T$ is a forget bag for $\pi$, $G$ has a connected matching of weight $w$ if and only if $\{(\{\pi\}, w)\}$ is the unique element of $f_{r'}(\{\pi\}, \emptyset)$, where $r'$ is the child bag of $r$.
    In the final step of the algorithm, we return the maximum weight obtained between all choices of $\pi$.
    
    As to the running time of the dynamic programming algorithm, note that, for each choice of $\pi$ we can compute all entries of $f_x$ in time bounded by {$2^{|B_x|} \cdot \sum_{i = 0}^{|B_x|} \binom{|B_x}{i}2^{\omega i}t^\bigO{1} \leq (1 + 2^{\omega})^{2t}t^\bigO{1}$}; the term $2^{\omega i}$ corresponds to the time needed to execute the algorithm of Theorem~\ref{thm:reduce} for an entry where $|S| = i$ and the term $2^{|B_x|}$ comes from all possible choices of $U$.
    Forget nodes can be computed in the same time since we make the same number $\bigO{t}$ fewer calls to Theorem~\ref{thm:reduce} for each entry.
    Finally, leaves can be solved in constant time and table $f_x$ for a join node $x$ can be computed in {$2^{|B_x|} \cdot \sum_{i = 0}^{|B_x|} \binom{|B_x}{i}2^{\omega i + i}t^\bigO{1} \leq (1 + 2^{2\omega})^{2t}t^\bigO{1}$} time; in this case, the $2^{\omega i + i}$ terms comes from the $2^i$ choices for $Y$, each of which requires one invocation of Theorem~\ref{thm:reduce}.
    Given that we have $\bigO{nt}$ nodes in a nice tree decomposition, our dynamic programming algorithm can be computed in $\bigO{nt}$ times the cost of the most expensive nodes, which are the join nodes, totaling the required $2^\bigO{t}n^\bigO{1}$ time.
    Since we have to apply it for each $\pi \in V(G)$, the entire algorithm runs in $2^\bigO{1}n^\bigO{1}$ time.
\end{proof}

\section{Conclusions and future works}\label{sec:conclusions}

Motivated by previous works on weighted $\mathscr{P}$-matchings, such as \pname{Weighted Induced Matching}~\cite{klemz2022}\cite{panda2020} and \pname{Weighted Acyclic Matching}~\cite{dieter2019}, in this paper we introduced and studied the \pname{Weighted Connected Matching} problem.

We begin our investigation on the complexity of the problem by imposing restrictions on the input graphs and weights.
In particular, we show that, the problem is \NPc\ on planar bipartite graphs and bipartite graphs of diameter 4 for binary weights, and on planar subcubic graphs and starlike graphs when weights are restricted to $\{-1, 1\}$.
On the positive side, we present polynomial time algorithms for \pname{Maximum Weight Connected Matching} on chordal graphs with non-negative weights, and on trees and subcubic graphs with arbitrary weights.
Our last contributions are in parameterized complexity, where we show that the problem admits a $2^\bigO{t}$ time algorithm when parameterized by treewidth, but does not admit a polynomial kernel when parameterized by vertex cover and the minimum required weight even with binary weights unless $\NP \subseteq \coNP/\poly$.

Possible directions for future work include determining the complexity of the problem for different combinations of graph classes and allowed edge weights.
In particular, we would like to know the complexity of \pname{Weighted Connected Matching} for diameter $3$ bipartite graphs when weights are non-negative, and for planar graphs of maximum degree at least 3 under the same constraint.
Other graph classes of interest include cactus graphs and block graphs, both with and without weight restrictions.
We are also interested in the parameterized complexity of the problem.
In terms of natural parameterizations, we see two possible directions: parameterizing by the number of edges in the matching or by the weight of the matching; while we have some negative kernelization results for these parameters, tractability is still an open question.
Other possibilities include the study of other structural parameterizations, with the main open question being tractability for the cliquewidth parameterization.
Finally, investigating other $\mathscr{P}$-matching problems, like \pname{Disconnected Matching} and \pname{Uniquely Restricted Matching} is also an interesting venue.
While most unweighted $\mathscr{P}$-matching problems are already {\NPH}, their weighted versions may be tractable for relevant graph classes.

\bibliographystyle{splncs04}
\bibliography{refs}

\end{document}